\documentclass[12pt]{article}

\usepackage[centertags]{amsmath}
\usepackage{amsmath}
\usepackage{amsfonts}
\usepackage{amssymb}
\usepackage{amsthm}
\usepackage{newlfont}
\usepackage{txfonts}
\usepackage{lmodern}
\usepackage{epsfig}
\usepackage{amscd}
\usepackage{graphicx}

\usepackage[usenames,dvipsnames]{color}

\newcommand{\RR}{{\mathbb R}}

\newcommand{\CC}{{\mathbb C}}
\newcommand{\calD}{{\cal D}}
\newcommand{\beq}{\begin{equation}}
\newcommand{\eeq}{\end{equation}}
\newcommand{\ba}{\begin{array}}
\newcommand{\ea}{\end{array}}
\newcommand{\bea}{\begin{eqnarray}}
\newcommand{\eea}{\end{eqnarray}}

\renewcommand{\Re}{\mathop{\mathrm{Re}}}
\renewcommand{\Im}{\mathop{\mathrm{Im}}}

\newtheorem{theorem}{Theorem}[section]

\newtheorem{lemma}{Lemma}[section]
\newtheorem{proposition}{Proposition}[section]
\newtheorem{definition}{Definition}[section]

\newtheorem{remark}{Remark}[section]

\newcounter{one}
\newcounter{two}
\begin{document}
\begin{center}

{\large   \bf The Cauchy problem for the Pavlov equation} 
 
\vskip 15pt

{\large  P. G. Grinevich$^{1}$,  P. M. Santini$^{2}$ and D. Wu$^{3,\S}$}

\vskip 8pt

{\it 
$^1$ Landau Institute for Theoretical Physics, Chenogolovka, Russia,\\
Lomonosov Moscow State University, Russia and\\
Moscow Physical Thechnical Institute, Dolgoprudnyi, Russia. 

\smallskip

$^2$ Dipartimento di Fisica, Universit\`a di Roma "La Sapienza" and\\
Istituto Nazionale di Fisica Nucleare, Sezione di Roma 1\\
Piazz.le Aldo Moro 2, I-00185 Roma, Italy

\smallskip

$^3$ Institute of Mathematics, Academia Sinica, 
Taipei, Taiwan}

\vskip 5pt

$^\S$e-mail: {\tt wudc@math.sinica.edu.tw}

\vskip 5pt

{\today}

\end{center}

\begin{abstract}
Commutation of multidimensional vector fields leads to integrable nonlinear dispersionless PDEs arising in various problems of mathematical physics and intensively studied in the recent literature. This report is aiming to solve  the scattering and inverse scattering problem for integrable dispersionless PDEs, recently introduced just at a formal level, concentrating on the prototypical example of the Pavlov equation, and to justify an existence theorem for global bounded solutions of the associated Cauchy problem with small data.
\end{abstract}

\section{Introduction}
Integrable soliton equations, like the Korteweg - de Vries \cite{KdV}, the Nonlinear Scrh\"odinger \cite{ZSNLS} equations and their integrable $(2+1)$ dimensional generalizations, the Kadomtsev - Petviashvili \cite{KP} and Davey - Stewartson \cite{DS} equations respectively, play a key role in the study of waves propagating in weakly nonlinear and dispersive media. The Inverse Spectral Transform (IST) method, introduced by Gardner, Green, Kruskal and Miura \cite{GGKM}, is the spectral method allowing one to solve the Cauchy problem for such PDEs, predicting that a localized disturbance evolves into a number of soliton pulses + radiation, and solitons arise as an exact balance between nonlinearity and dispersion \cite{ZMNP},\cite{AS},\cite{CD},\cite{AC}. There is another important class of integrable PDEs, the so-called dispersionless PDEs (dPDEs), or PDEs of hydrodynamic type, arising in various problems of Mathematical Physics and intensively studied in the recent literature (see, f.i., in the multidimensional context, \cite{ZS,Zak,KG,Kri0,Kri1,Taka1,TT1,TT2,Zakharov,Kri2,TT3,DuMa1,DuMa2,DMT,DT,K-MA-R,MA-S,G-M-MA,Pavlov,Duna,Ferapontov,NNS,BK,KM1,KM2}). The class of integrable dPDEs includes relevant examples, like the dispersionless Kadomtsev - Petviashvili (dKP) equation \cite{Lin},\cite{Timman},\cite{ZK}, describing the evolution of weakly nonlinear, nearly one-dimensional waves in Nature, in the absence of dispersion and dissipation \cite{Lin}, \cite{Timman}, \cite{ZK}, \cite{MS8}, the first and second heavenly equations of Plebanski \cite{Plebanski}, relevant in complex gravity, and the dispersionless 2D Toda (or Boyer-Finley) equation \cite{FP,BF}, whose elliptic and hyperbolic versions are relevant in twistor theory \cite{BF,GD} as integrable Einstein - Weyl geometries \cite{H,J,Ward}, and in the ideal Hele-Shaw problem \cite{MWZ,WZ,KMZ,LBW,MAM}.

Since integrable dPDEs arise from the condition of commutation $[L ,M]=0$ of pairs of one-parameter families of vector fields, implying the existence of common zero energy eigenfunctions (elements of the common kernel):
\beq\label{Lax_pair}
[L,M]=0~~\Rightarrow~~L\psi=M\psi=0,~~j=1,2,
\eeq
they can be in an arbitrary number of dimensions \cite{ZS}, unlike the soliton PDEs. In addition, due to the lack of dispersion, these multidimensional PDEs may or may not exhibit a gradient catastrophe at finite time. To investigate integrable dPDEs, a novel IST for vector fields, significantly different from that of soliton PDEs, has been recently constructed in \cite{MS0,MS1,MS2}, just at a formal level, i) to solve their Cauchy problem, ii) obtain the longtime behavior of solutions, iii) costruct distinguished classes of exact implicit solutions, iv) establish if, due to the lack of dispersion, the nonlinearity of the dPDE is ``strong enough'' to cause the gradient catastrophe of localized multidimensional disturbances, and v) to study analytically the breaking mechanism \cite{MS0,MS1,MS2,MS3,MS4,MS5,MS6,MS7,MS8,MS9,MS10}. 

It is important to remark that this novel IST is based on some critical assumptions, like existence of analytic eigenfunctions. In soliton theory we know that, in contrast with 1+1 systems, the relevant eigenfunctions for many 2+1 PDEs (like KP\Roman{two}) are not analytic \cite{ABF}, and the inverse problem is formulated as a $\bar\partial$-problem. But the methods used in soliton theory for proving the existence of the relevant eigenfunctions fail in the dispersionless case, since the corresponding operators are unbounded. In addition, since the Lax operators are vector fields, the kernel space is a ring, and the inverse problem is intrinsically 
nonlinear. Al last, the dispersionless theory lacks of explicit regular localized solutions (solitons or lumps do not exist), and gradient catastrophes of different nature may occur at finite time. 

For all these reasons, it is clearly important to make the IST for vector fields rigorous (even more important than for the case of soliton PDEs); and this is the main goal of this work.

To do that, we choose, as illustrative example, the simplest integrable nonlinear dPDE available in the literature, the so-called Pavlov equation \cite{Pavlov}, \cite{Ferapontov}, \cite{Duna} 
\beq 
\label{Pavlov}
\ba{l}
v_{xt}+v_{yy}+v_xv_{xy}-v_yv_{xx}=0,~~v=v(x,y,t)\in\RR,~~x,y,t\in\RR, 
\ea
\eeq
arising in the study of integrable hydrodynamic chains  \cite{Pavlov}, and in Differential Geometry as a particular example of Einstein - Weyl metric \cite{Duna}. It was first derived in \cite{Duna1} as a conformal symmetry of the second heavenly equation. 

As it was pointed out to the authors \cite{ZKpr1},
the terms $v_{xt}+v_xv_{xy}-v_yv_{xx}$ in equation (\ref{Pavlov}) are in common 
(up to the interchange of $x$ and $y$) with the zero pressure Prandtl's equation for the potential $\Phi$ \cite{EE}: 
\beq 
\label{Prandtl}
\Phi_{yt}-\Phi_{yyy}+\Phi_{y}\Phi_{xy}-\Phi_{x}\Phi_{yy}=0.
\eeq
The main difference between these two equations is that the friction term of the Prandtl's equation is replaced by the diffraction term of the Pavlov equation. While the zero-pressure Prandtl's equation with suitable boundary conditions gives rise to blow-up at finite time  \cite{EE}. We prove in this paper that localized and sufficiently small initial data for Pavlov equation remain smooth at all times. 

The inviscid  Prandtl's equation 
\beq 
\label{Prandtl2}
\Phi_{yt}+\Phi_{y}\Phi_{xy}-\Phi_{x}\Phi_{yy}=0
\eeq
can be linearized using some partial Legendre transformation, and it also shows formation of singularities at finite time (unpublished result by V.E. Kuznetsov \cite{Kpr1}).

Equation (\ref{Pavlov}) arises as the commutativity condition (\ref{Lax_pair}) of the following pair of vector fields \cite{Duna}
\beq 
\label{Lax_Pavlov}
\ba{l}
L\equiv \partial_y+(\lambda +v_x)\partial_x, \\
M\equiv \partial_t+(\lambda^2+\lambda v_x-v_y)\partial_x ,
\ea
\eeq
and is the $u=0$ reduction of the following integrable system of dispersionless PDEs \cite{MS2}
\beq
\label{dKP-system}
\ba{l}
u_{xt}+u_{yy}+(uu_x)_x+v_xu_{xy}-v_yu_{xx}=0,         \\
v_{xt}+v_{yy}+uv_{xx}+v_xv_{xy}-v_yv_{xx}=0, 
\ea
\eeq
describing the most general integrable Einstein - Weyl metric \cite{Dunajj}, \cite{DunaFer}. This system reduces instead, for $v=0$, to the celebrated dKP equation:
\beq\label{dKP}
\ba{l}
u_{xt}+u_{yy}+(uu_x)_x=0,~~u=u(x,y,t)\in\RR,~~x,y,t\in\RR, 
\ea
\eeq
the simplest prototype integrable model for the study of wave breaking in multidimensions \cite{MS4},\cite{MS9}.      

Let us point out that, although the linearized versions of the Pavlov and dKP equations coincide, the formal IST predicts a regular dynamics for the Pavlov equation, and the gradient catastrophe
at finite time for the dKP equation.

In our paper we prove the following result:

\begin{theorem}
Suppose that $v_0(x,y)$ is a Schwartz function with compact support and satisfies a {\bf small norm condition} (see Definition~\ref{def:small-data}). Then the IST method provides us with a real function $v(x,y,t)$ such that $v(x,y,0)=v_0(x,y)$, the functions $\partial_x v(x,y,t)$, $\partial_y v(x,y,t)$, 
$\partial_x^2v(x,y,t)$, $\partial_y\partial_xv(x,y,t)$, 
$\partial_t\partial_xv(x,y,t)$, $\partial_y^2v(x,y,t)$ lie in $C(\RR\times\RR\times \RR^+)\cap L^\infty(\RR\times\RR\times \RR^+)$ and satisfy the Pavlov equation (\ref{Pavlov}).
\end{theorem}
\begin{remark}
The behavior of $\partial_t v(x,y,t)$ at $t=0$  requires an extra investigation.
\end{remark}
\vskip10pt 

Since the realization of the scheme described above requires a rather big amount of technical work, including estimates on the behavior of the integral equations kernels, to make our text more transparent, we moved the proofs of the analytic estimates to the last section of our paper.

The authors would like to dedicate this paper to the memory of S. V. Manakov who successfully devoted the last period of his life to the construction of the IST method for vector fields, and to its applications to the theory of integrable dispersionless PDEs in multidimensions. 

\noindent

\section{The Inverse Scattering Transform: a short summary}

We find it convenient to summarize here the basic formal steps associated with this novel IST for the vector field $L$ in (\ref{Lax_Pavlov}), allowing one to solve the Cauchy problem for the Pavlov equation \cite{MS0,MS1,MS3}, whose rigorous aspects will be investigated in the following sections.

\textbf{The Direct Problem} In our paper we always assume that $v(x,y)$ is a real-valued function. In analogy with the IST for KP\Roman{one} equation (whose Lax operator in the non-stationary Schr\"odinger operator, see \cite{Manakov1}, \cite{FA}), we make essential use of two sets of eigenfunctions -- the real Jost eigenfunctions $\varphi_{\pm}(x,y,\lambda)$, $\lambda\in\RR$,  and the complex-analytic in $\lambda$ ones: $\Phi^{+}(x,y,\lambda)$, $\Im \lambda\ge0$;   $\Phi^{-}(x,y,\lambda)$, $\Im \lambda\le0$ 
\begin{align}
& L\varphi_{\pm}(x,y,\lambda)=0, \ \  L\Phi^{\pm}(x,y,\lambda)=0,\\
& \varphi_{\pm}(x,y,\lambda)\rightarrow x-\lambda y \ \ \mbox{as} \ \ y\rightarrow\pm\infty. \label{def-phi}
\end{align}
The direct spectral transform consists of two steps
\begin{itemize}
\item Using the real Jost eigenfunctions we construct the scattering data $\sigma(\xi,\lambda)$. 
\item Using the complex-analytic eigenfunctions we construct the spectral data $\chi(\xi,\lambda)$ through the scattering data.
\end{itemize}

{\bf Step 1.} For real $\lambda$, all eigenfunctions of $L$ have the following property: they are constant on the trajectories of the following ODE:  
\beq\label{ODE}
\frac{dx}{dy}=\lambda+v_x(x,y)
\eeq
defining the characteristics of $L$. Indeed, if the potential $v$ is sufficiently regular and well-localized, the solution of the Cauchy problem 
$x(y_0)=x_0$  for the ODE (\ref{ODE}) exists unique globally in the (time) variable $y$, with the following free particle asymptotic behavior
\beq\label{asympt_ODE}
x(y)\to \lambda y+x_{\pm}(x_0,y_0,\lambda),~~y\to\pm\infty.
\eeq 
The asymptotic positions $x_{\pm}(x_0,y_0,\lambda)$ are obviously constant when the point $(x_0,y_0)$ moves along trajectories. Therefore  $x_{\pm}(x_0,y_0,\lambda)$ are solutions of the vector field equation
$$
[\partial_{y_0}+(\lambda +v_{x_0}(x_0,y_0) )\partial_{x_0}]  x_{\pm}(x_0,y_0,\lambda)=0.
$$
Due to (\ref{asympt_ODE}) we have
$$
x_{\pm}(x_0,y_0,\lambda)\rightarrow x_0 -\lambda y_0 \ \ \mbox{as} \ \ y_0\rightarrow\pm\infty,
$$
therefore they coincide with the real Jost eigenfunctions
\beq\label{asympt_ODE_Pavlov}
\varphi_{\pm}(x_0,y_0,\lambda)=x_{\pm}(x_0,y_0,\lambda).
\eeq 

\begin{definition} 
\label{def:sigma}
Denote  by $\sigma(\xi,\lambda)$ the classical time-scattering datum, connecting the asymptotic behavior of the solutions at $y\rightarrow+\infty$ and at  $y\rightarrow-\infty$
$$
x_{+}(x_0,y_0,\lambda)=x_{-}(x_0,y_0,\lambda)+\sigma(x_{-}(x_0,y_0,\lambda),\lambda),
$$
\end{definition}
therefore
\beq
\label{def-S}
\varphi_{+}(x,y,\lambda)\rightarrow x-\lambda y + \sigma(x-\lambda y,\lambda) \ \ \mbox{as} \ \ y\rightarrow-\infty.
\eeq
{\bf Step 2. }The problem of existence for complex (analytic) eigenfunctions $\Phi^{\pm}$ of a vector field is usually highly nontrivial, and in all previous works by Manakov and Santini was only postulated and motivated by the analyticity properties of the Green's functions of the undressed vector fields. In our paper we present a proof based on the following observation:

For $\lambda\in\mathbb C/\mathbb R$ , by the change of variables $z=x-\lambda y$, $\bar z=x-\bar\lambda y$, the Lax equation $L\Phi(x,y,\lambda)=0$ can be transformed into a linear Beltrami equation and can be solved. Moreover, we do not have to assume, at this stage, that the potential $v(x,y)$ has small norm.

We show below that the limiting functions $ \Phi^{\pm}(x,y,\lambda)=  \Phi(x,y,\lambda\pm i 0)$, $\lambda\in\RR$ are also well-defined. Both real Jost eigenfunctions $\varphi_{\pm}(x,y,\lambda)$ enumerate the trajectories of our vector field, therefore any eigenfunction of $L$ for $\lambda\in\RR$ can be 
represented as a function either of  $\varphi_{+}(x,y,\lambda)$ or  $\varphi_{-}(x,y,\lambda)$, and we have:
\begin{align}
\label{E:intro-1}
&\Phi^{-}(x,y,\lambda) 
 =\varphi_{-}(x,y,\lambda)+\chi_{-}(\varphi_{-}(x,y,\lambda),\lambda) = \varphi_{+}(x,y,\lambda)+\chi_{+}(\varphi_{+}(x,y,\lambda),\lambda) 
\nonumber
\\
&\Phi^+(x,y,\lambda) =\overline{\Phi^-(x,y,\lambda)}.
\end{align}
defining the spectral data $\chi_{\pm}(\xi,\lambda)$.

Assuming that the small $\lambda_I=\Im \lambda$ behaviour be sufficiently good, we see that, for $\lambda_I\to 0$, the eigenfunction $\Phi(x,y,\lambda)$ is almost constant on the trajectories of the vector field $\hat L\equiv \partial_y+(\lambda_R+v_x)\partial_x$; these trajectories are straight lines $\Re z=\mathrm{const}$ outside the support of $v(x,y)$ and connect the lines $\Re z=\xi$ and $\Re z =\xi+\sigma(\xi,\lambda)$ as they go from $-\infty$ to $+\infty$ (see Fig~\ref{fig1}). 

\begin{figure}[h]
\begin{center}
\includegraphics[height=6cm]{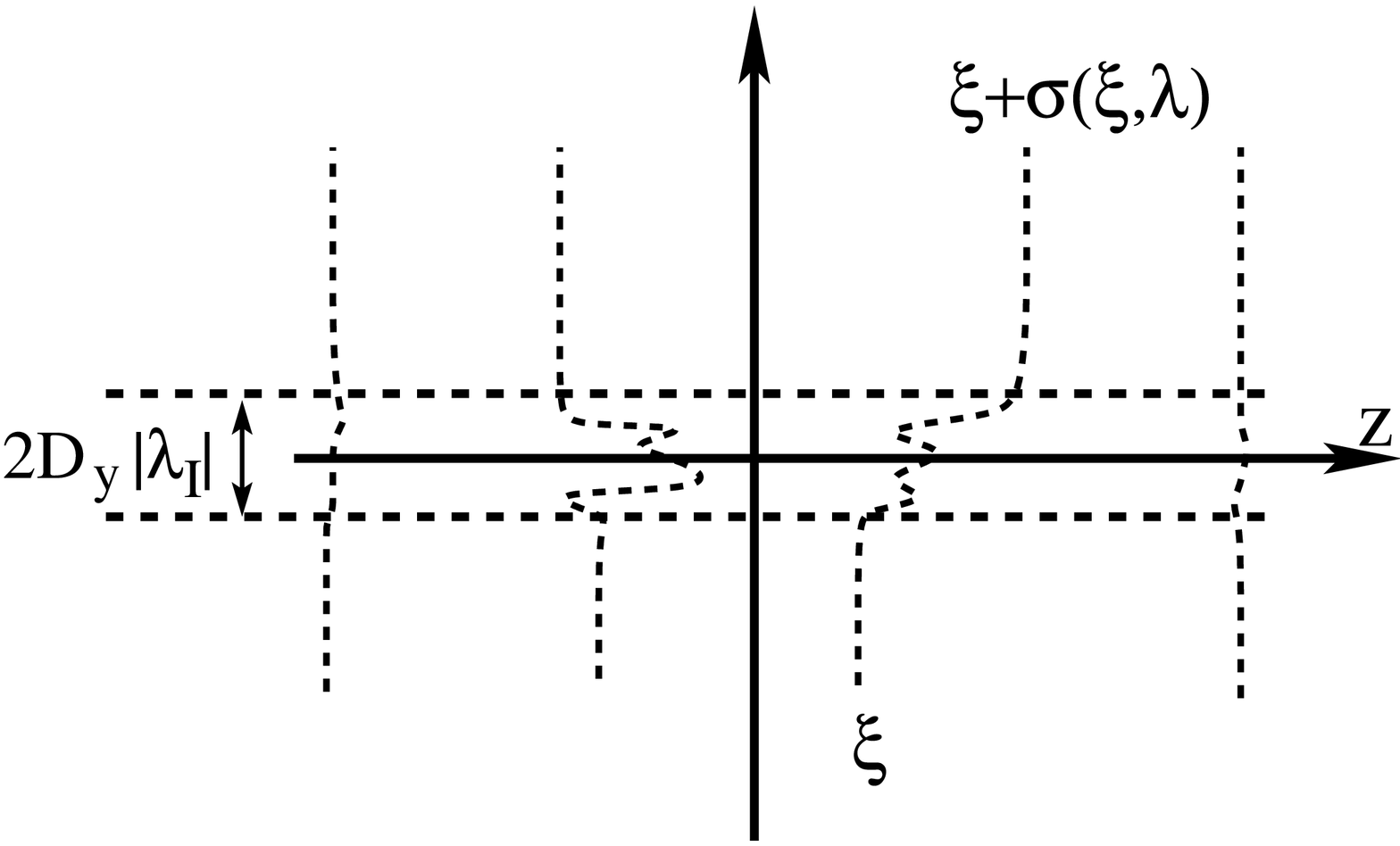}
\end{center}
\caption{\label{fig1} The trajectories of the vector field for $\Im\lambda\ll 1$.}
\end{figure}

Assume now that $\lambda_I<0$, $|\lambda_I|\ll 1$;  then  $\Phi^{-}(x,y,\lambda)$ is holomorphic in $z$ outside a small neighbourhood of $\mathbb R$:  $\Phi^{-}(x,y,\lambda)= \hat \Phi(z,\lambda)$ and, due to the almost constant behavior on the trajectories: 
\begin{equation}
\label{RH_shift}
 \hat\Phi(\xi-i\epsilon,\lambda) \sim \hat\Phi(\xi+\tilde\sigma(\xi,\lambda)+i\epsilon,\lambda). 
\end{equation}
In the limit $\lambda_I\rightarrow 0-$ we have
\begin{align}
& \hat\Phi(\xi-i0,\lambda)  = \Phi^{-}(x,y,\lambda), \ \ y < -D_y,\\
& \hat\Phi(\xi+i0,\lambda)  = \Phi^{-}(x,y,\lambda), \ \ y > D_y,\\
\end{align}
therefore equation~(\ref{E:intro-1}) implies
$$
\hat\Phi(\xi-i0,\lambda) = \xi + \chi_{-}(\xi,\lambda), \ \   \hat\Phi(\xi+i0,\lambda) = \xi + \chi_{+}(\xi,\lambda).
$$
Hence the spectral data $\chi_{\pm}(\xi,\lambda)$
of the Pavlov equation satisfy the shifted Riemann-Hilbert (RH)  problem 
\beq\label{E:shift-intro}
\begin{split}
\sigma(\xi,\lambda)+\chi_{+}(\xi+\sigma(\xi,\lambda),\lambda)-\chi_{-}(\xi,\lambda)=0,\quad \xi\in \RR,\\
\hskip1in\partial_{\bar\xi}\chi=0\ \textit{ for } \xi\in\CC^\pm,\quad\quad\quad
\chi\to 0\ \textit{ as } |\xi|\to\infty.
\end{split}
\eeq
Equation (\ref{E:shift-intro}) defines the spectral data $\chi_{\pm}(\xi,\lambda)$ in terms of the scattering data $\sigma(\xi,\lambda)$. No small norm assumption is required also at this step.

\textbf{Evolution of the spectral data}. The evolution of the scattering and spectral data, following from the asymptotics (\ref{def-phi}) and (\ref{def-S}), is given by the explicit formula \cite{MS2,MS3}:
\beq
\label{eq:spectral-evolution}
\ba{l}
\sigma(\xi,\lambda,t)=\sigma(\xi-\lambda^2 t,\lambda,0),\\
\chi_{\pm}(\xi,\lambda,t)=\chi_{\pm}(\xi-\lambda^2 t,\lambda,0),\
\ea
\eeq
implying that, from the eigenfunctions $\varphi_{\pm},\Phi^{\pm}$ of $L$, one can constructs the common eigenfunctions $\psi_{\pm},\Psi^{\pm}$ of $L$ and $M$ through the formulas
\beq\label{phi_Phi}
\ba{l}
{\psi_{\pm}}(x,y,t,\lambda)={\varphi_{\pm}}(x,y,t,\lambda)-\lambda^2 t,~~
\Psi^{\pm}(x,y,t,\lambda)=\Phi^{\pm}(x,y,t,\lambda)-\lambda^2 t,
\ea
\eeq
connected through equations 
\begin{align}
\label{E:intro-1-bis}
&\Psi^{-}(x,y,t,\lambda) 
 =\psi_{-}(x,y,t,\lambda)+\chi_{-}(\psi_{-}(x,y,t,\lambda),\lambda) = \psi_{+}(x,y,\lambda)+\chi_{+}(\psi_{+}(x,y,\lambda),\lambda) 
\nonumber
\\
&\Psi^+(x,y,t,\lambda) =\overline{\Psi^-(x,y,t,\lambda)}.
\end{align}

\textbf{The inverse problem} The reconstruction of the real eigenfunction $\psi_-$ at time $t$ from the spectral data $\chi_-$ is provided by the solution of the nonlinear integral equation 
\beq\label{inversion1}
\psi_-(x,y,t,\lambda)-H_{\lambda}\chi_{-I}\big(\psi_-(x,y,t,\lambda),\lambda\big)+\chi_{-R}\big(\psi_-(x,y,t,\lambda),\lambda\big)=x-\lambda y-\lambda^2 t,
\eeq
where $\chi_{-R}$ and $\chi_{-I}$ are the real and imaginary parts of $\chi_-$ , and $H_{\lambda}$ is the Hilbert transform operator wrt $\lambda$
\beq
H_{\lambda}f(\lambda)=\frac{1}{\pi} \fint\limits_{-\infty}^{\infty}\frac{f(\lambda')}{\lambda-\lambda'}d\lambda' .
\eeq
We remark that, since $\chi_-(\xi,\lambda)$ is analytic wrt $\xi$ in the lower half-plane, its real and imaginary parts satisfy the relation $\chi_{-R}-H_{\xi}\chi_{-I}=0$. Equation (\ref{inversion1}) expresses the fact that the RHS of (\ref{E:intro-1-bis}) for $\Psi^{-}$ is the boundary value of a function analytic in  $\lambda$ in the lower half-plane.

Once $\psi_{-}$ is reconstructed from $\chi_-$ solving the nonlinear integral equation (\ref{inversion1}), equation (\ref{E:intro-1-bis}) gives $\Psi^{\pm}$, and $v$ is finally reconstructed from: 
\beq\label{inv_u_1}
\ba{l}
v(x,y)=-\displaystyle\lim_{\lambda\to\infty}{\left(\lambda[\Psi^-(x,y,\lambda)-(x-\lambda y-\lambda^2 t]\right)}, 
\ea
\eeq
or, better, as we shall see, from
\begin{equation}
\label{eq:v-rieamann-1}
v(x,y,t)= -\frac{1}{\pi} \int\limits_{\RR} \chi_{-I}(\psi_{-}(x,y,t,\zeta),\zeta) d\zeta.
\end{equation}

\begin{remark}The main difficulty associated with the direct problem is in the proof of the existence of the analytic eigenfunction and of its limits on the real $\lambda$ axis from above and below. While such a proof can be made in the Pavlov case, see subsection~\ref{S:complex}, in the dKP case the existence of the analytic eigenfunctions is proven, at the moment, only sufficiently far from the real $\lambda$ axis \cite{GS}. We also remark that, soon after the formulation of the direct problem through the RH problem (\ref{RH_shift}) \cite{MS0}, an alternative integral equation, obtained taking the Fourier transform of (\ref{RH_shift}), was also suggested \cite{MS1},\cite{MS2}. It turns out that, while the construction of the spectral data from the scattering data through the RH problem with shift (\ref{RH_shift}) does not present difficulties, the construction that makes use of the integral equation in Fourier space requires additional effort, due to the bad behavior of its kernel, and will not be considered in this paper. 
\end{remark}

\begin{remark}A second inverse problem, a nonlinear RH (NRH) problem on the real line, was also introduced at a formal level \cite{MS0,MS1,MS2}, and intensively used i) to study the longtime behavior of the solutions of the target dPDE \cite{MS4,MS5,MS6}; ii) to detect if a localized initial disturbance evolving according to such a PDE goes through a gradient catastrophe at finite time (f.i., no gradient catastrophe for the second heavenly equation \cite{MS1,MS5} and for the Pavlov equation \cite{MS3} was found, while a gradient catastrophe was indeed found for the dKP \cite{MS6} and for the dispersionless 2D Toda \cite{MS6} equations); iii) to investigate analytically the wave breaking mechanism of such multidimensional waves \cite{MS4,MS6}; iv) to construct classes of RH data giving rise to exactly solvable NRH problems, and to distinguished exact implicit solutions of the dispersionless PDEs through an algorithmic approach \cite{MS7,BDM,MS4,MS5,MS6}; v) to detect integrable differential reductions of the associated hierarchy of PDEs \cite{Bogdanov1,Bogdanov2}, like the Dunajski interpolating equation $v_{xt}+v_{yy}+cv_xv_{xx}+v_xv_{xy}-v_yv_{xx}=0$ \cite{Dunajj}, an integrable PDE interpolating between the dKP and the Pavlov equations, corresponding to the reduction $u=c v_x$ of system (\ref{dKP-system}). The rigorous aspects of such a NRH inverse problem, as well as the connections with the above inverse problem, will also be investigated in a subsequent paper. 
\end{remark}


\section{Direct spectral transform}

\subsection{The real eigenfunction}\label{S:real-eigenfunction}

Throughout this paper, 
\begin{eqnarray*}
&&\mathfrak S_{x,y} =\{f:\mathbb R^2\to\mathbb R|\textit{$\ \|f\|^{(\mu,\nu)}_{( k,h)}=|(1+|x|)^k(1+| y|)^h\partial_x^\mu\partial_y^\nu f|_{L^\infty}<\infty$}\},\\
&&L^p(\mathbb R, d\lambda) =\{f:\mathbb R\to\mathbb C|\ \|f\|_{L^p(\mathbb R, d\lambda)}= (\int_{\mathbb R}| f(\lambda)|^pd\lambda)^\frac 1p<\infty\},\\
&&W^{k,2}(\mathbb R, d\lambda) =\{f:\mathbb R\to\mathbb C|\textit{   $\|f\|_{W^{k,2}(\mathbb R, d\lambda)}=\sum_{j=0}^k\left\|\partial_\lambda^jf\right\|_{L^2(\mathbb R, d\lambda)}<\infty$}\}.
\end{eqnarray*} 

We shall also use the Sobolev spaces with additional weights:

$$
W^{k,2}(\mathbb R, d\lambda)(\mu) =\{f:\mathbb R\to\mathbb C|\textit{   $\|f\|_{W^{k,2}(\mathbb R, d\lambda)(\mu)}=\sum_{j=0}^k \mu^k \left\|\partial_\lambda^jf\right\|_{L^2(\mathbb R, d\lambda)}<\infty$}\}, \ \ \mu>0.
$$
For all positive $\mu$ these norms are equivalent, but in some situations it is necessary to choose an appropriate 
$\mu$ to guaranty the contraction property for our integral operators.

In our paper we assume that the potential $v(x,y)$ has compact support in $x$, $y$. We expect that these constraints are not critical and can be weakened (for example, it should be enough to assume that the potential decays sufficiently fast as $x^2+y^2\rightarrow\infty$), but it may require a serious additional analytic work. To be more precise, let $D_x$, $D_y$ be a pair of positive numbers, $n>D_y$, and 
$v\in\mathfrak G_{x,y}$ such that
\begin{equation}\label{E:v-supp}
v(x,y)=0 \ \ \mbox{for} \ \ |x|>D_x \ \ \mbox{or} \ \  |y|>D_y.
\end{equation}
The real eigenfunctions $\varphi_{\pm}(x,y,\lambda)$ for the Pavlov equation are defined by the solution of the boundary value problem: for each fixed $\lambda\in\RR$, \cite{MS5}
\begin{eqnarray} \partial_y\varphi_{\pm}+\left(\lambda+v_x\right)\partial_x\varphi_{\pm}=0,\quad&&\textit{ for $x,\,y\in\mathbb R$, }\label{E:pavlov-lax}\\
\varphi_{\pm}-\xi \to 0,\quad && \textit{ as $y\to \pm\infty$},\label{E:pavlov-lax-bdry}
\end{eqnarray}
where 
\begin{equation}
\xi=x-\lambda y.
\end{equation}
\begin{lemma}\label{L:pde-system}
Suppose $v\in\mathfrak S_{x,y}$ satisfying (\ref{E:v-supp}). The real eigenfunctions $\varphi_{\pm}$ exists uniquely and $\varphi_{\pm}-(x-\lambda y)$ are smooth bounded functions.
\end{lemma} 
\begin{proof} The solvability and uniqueness of the boundary value problem of the first order partial differential equation (\ref{E:pavlov-lax}), (\ref{E:pavlov-lax-bdry})  can be derived by solving the ordinary differential equation
\beq
\frac {d x}{dy}=\lambda+v_x(x,y), \ \ x=x(y;x_0,y_0,\lambda),  \ \ x(y_0;x_0,y_0,\lambda)=x_0, 
\eeq
or, equivalently, 
\beq
\label{eq:def-h}
\frac {d h}{dy}=v_x(h+\lambda y,y), \ \ h=h(y;\xi_0,y_0,\lambda), \ \ h(y_0;\xi_0,y_0,\lambda)=\xi_0=x_0-\lambda y_0, 
\eeq
where
$$
h(y) = x(y) -\lambda y.
$$

Using the Picard iteration method on the integral equation defining the solution (see, for example, \cite{Arnold})
\beq
h(y;\xi_0,y_0,\lambda) =h_0 +\int_{y_0}^y v_x\big(h(y';\xi_0,y_0,\lambda)+ \lambda y',y'\big)dy'
\label{E:h}
\eeq
one shows that $x_{\pm}(x_0,y_0,\lambda)=h(\pm n;x_0-\lambda y_0,y_0,\lambda)$ are smooth functions, 
$h(\pm n;x_0-\lambda y_0,y_0,\lambda)-h_0$ are also bounded. Here we used the fact that $h(y;x_0-\lambda y_0,y_0,\lambda)$ are constant in $y$ in the regions $y\ge n$, $y\le -n$ due to the compact support of $v(x,y)$.

We see, that 
$$
\sigma(\xi_0,\lambda)=h(n;\xi_0,-n,\lambda)-\xi_0=\int_{-\infty}^{\infty} v_x\big(h(y';\xi_0,-\infty,\lambda)+ \lambda y',y'\big)dy'
$$
is also a regular function and the map $\xi\rightarrow\xi+\sigma(\xi,\lambda)$ is regularly invertible for all $\lambda$. We do not require the small norm assumption at this step.

\end{proof}

For simplicity and convenience, we will use the following agreement: $C$  denotes a constant, possibly dependent of 
$\|v\|^{(\mu,\nu)}_{(k,h)}$, but independent of $x$, $y$, $t$, and $\lambda$ throughout this paper. To construct the spectral data from the scattering data by solving the shifted Riemann-Hilbert problem, it is necessary to control the behavior of the scattering data and its derivatives for large $\lambda$. For solving the inverse problem we also need some estimates for large $\lambda$ and $\xi\sim\lambda^2$.

\begin{proposition}\label{L:direct-sigma-asy-00}
Suppose $v\in\mathfrak S_{x,y}$ such, that $v(x,y)\equiv 0$ for $|y|>D_y$. Let us define the following constants $B_k=b_k[v]$, $k=0,1,2,3$, $\hat B_k=\hat b_k[v]$, $k=0,1$:
\beq
\label{sigma-est-B1}
B_0=\int\limits_{-\infty}^{+\infty} \left[\max\limits_{x\in\RR} |v_x(x,y)|  \right] dy,
\eeq
\beq
\label{sigma-est-B2}
B_1=\exp\left[ \int\limits_{-\infty}^{+\infty} \left[\max\limits_{x\in\RR} |v_{xx}(x,y)|  \right] dy \right] -1,
\eeq
\beq
\label{sigma-est-B3}
B_2=\left[ \int\limits_{-\infty}^{+\infty} \left[\max\limits_{x\in\RR} |v_{xxx}(x,y)|  \right] dy \right] (1+B_1)^3,
\eeq
\begin{align}
\label{sigma-est-B4}
B_3&=\left[ \int\limits_{-\infty}^{+\infty} \left[\max\limits_{x\in\RR} |v_{xxx}(x,y)|  \right] dy \right] 3 (1+B_1)^2B_2+
\\
&+\left[ \int\limits_{-\infty}^{+\infty} \left[\max\limits_{x\in\RR} |v_{xxxx}(x,y)|  \right] dy \right] (1+B_1)^4,
\end{align}
\begin{align}
\label{sigma-est-B5}
\hat B_0&=\left[ \int\limits_{-\infty}^{+\infty} \left( \sqrt{  \int\limits_{-\infty}^{+\infty} |v_{x}(x,y)|^2 dx } \right) dy \right] \cdot\frac{1}{\sqrt{1-B_1}},
\end{align}
\begin{align}
\label{sigma-est-B6}
\hat B_1&=\left[ \int\limits_{-\infty}^{+\infty} \left( \sqrt{  \int\limits_{-\infty}^{+\infty} |v_{xx}(x,y)|^2 dx } \right) dy \right] \cdot\frac{1+B_1}{\sqrt{1-B_1}}.
\end{align}
Then we have the following estimates on the scattering data:
\begin{equation}
\label{sigma-est-1}
|\sigma(\xi,\lambda)|\le B_0, \ \ |\sigma_{\xi}(\xi,\lambda)|\le B_1, 
\ \  |\sigma_{\xi\xi}(\xi,\lambda)|\le B_2, \ \  |\sigma_{\xi\xi\xi}(\xi,\lambda)|\le B_3.
\end{equation}
Moreover, if $B_1<1$, 
\begin{equation}
\label{sigma-est-2}
\|\sigma(\xi,\lambda)\|_{L^2(d\xi)}\le \hat B_0, \   \|\sigma_{\xi}(\xi,\lambda)\|_{L^2(d\xi)}\le \hat B_1\nonumber.
\end{equation}
\end{proposition}

The proof of Proposition~\ref{L:direct-sigma-asy-00}  is moved to the last Section. It is rather straightforward and is based on some standard estimates from the ODE theory,

\begin{definition}
\label{def:small-data}
A potential $v(x,y)$ satisfies {\bf the small norm condition} if the following inequalities are fulfilled:
\begin{enumerate}
\item $B_0\le \frac{1}{4}$,
\item $B_1\le \frac{1}{2}$,
\item $8 B_0+ 4 B_2 + 2\sqrt{2} \hat B_0 <\pi$,
\item $2B_1+\frac{1}{\pi}(64 B_1+ 16 \hat B_1) + \frac{1}{\pi}(8B_3+
16 B_2^2 + 56 B_1+ 16 B_1^2) \left(B_0+\frac{2}{\pi}[2 B_0 + \hat B_0] \right)<\tan\left(\frac{\pi}{8}\right)$.
\end{enumerate}
The meaning of the combinations of constants arising in this definition will be explained later.

\end{definition}

\begin{proposition}\label{L:direct-sigma-asy-0}
Suppose $v\in\mathfrak S_{x,y}$  satisfying (\ref{E:v-supp}) and $|\lambda|$ is sufficiently 
large. Let us introduce new variables
$$
{\mathring\lambda}=\frac{1}{\lambda}, \ \ \mathring{\xi} = \frac{\xi}{\lambda}.
$$
Then, for sufficiently small ${\mathring\lambda}$, the function 
$$
\mathring\sigma(\mathring{\xi},{\mathring\lambda}) = \sigma(\mathring\xi/{\mathring\lambda},1/{\mathring\lambda})/{\mathring\lambda}^2
$$
has the following properties:
\begin{enumerate}
\item It vanishes outside the interval $|\mathring{\xi}|\le D_y+|{\mathring\lambda}| D_x$.
\item It is smooth in both variables $\mathring{\xi}$, ${\mathring\lambda}$.
\end{enumerate}
As a corollary we obtain that there exists a collection of positive constants $C^{(\mu,k)}$, such that 

\begin{eqnarray}
\|\partial^k_{\lambda} \partial^{\mu}_{\xi}\sigma(\xi,\lambda)\|_{L^\infty}&<& \frac{C^{(\mu,k)}}{1+|\lambda |^{2+\mu+k}}, \ \ \mu\ge 0, \ \ k\ge 0, 
\label{E:l1-condition-lambda-tau-sigma-0}\\
\|\partial^k_{\lambda} \partial^{\mu}_{\xi}\sigma(\xi,\lambda)\|_{L^2(d\xi)}&<& \frac{C^{(\mu,k)}}{1+|\lambda |^{3/2+\mu+k}}, \ \ \mu\ge 0, \ \ k\ge 0,
\label{E:l1-condition-lambda-tau-sigma-1}
\end{eqnarray}
\end{proposition}

As usual, we move the proof to the last section of our paper.

\subsection{The complex eigenfunction}\label{S:complex}
In this section, we prove that there 
exists a unique eigenfunction $\Phi(x,y,\lambda)$ for each $\lambda\in\CC^\pm$. Moreover, $\Phi(x,y,\lambda)$ is holomorphic in $\lambda\in\CC^\pm$, its boundary values on $\RR$, denoted as $\Phi^\pm(x,y,\lambda)$, 
are well-defined and can be characterized by the shifted Riemann-Hilbert problem (\ref{E:shifted-RH}). 

For $\lambda\in \CC^\pm$, we introduce the following complex notations:
\begin{align}
\label{eq:complex1}
&z= x-\lambda y, 
&\bar{z}           &=  x-\bar \lambda y, \nonumber \\
&x= \frac{1}{\bar\lambda-\lambda}(\bar\lambda z -\lambda\bar z), 
&y                 &= \frac{1}{\bar\lambda-\lambda} (z - \bar z), \nonumber \\
&\partial_{\bar z}= -\frac{1}{\bar\lambda-\lambda} (\partial_y+\lambda\partial_x),  
&\partial_z        &=  \frac{1}{\bar\lambda-\lambda}(\partial_y+\bar\lambda\partial_x), \nonumber\\ 
&\partial_x= \partial_{z}+\partial_{\bar z}, 
&\partial_y        &=-(\bar\lambda \partial_{\bar z}+\lambda\partial_{z}).  \nonumber \\
\end{align}

So $W^{2,p}(dxdy)=W^{2,p}(dzd\bar z)=W^{2,p}$ for each $\lambda\in\CC^\pm$.

\begin{theorem}\label{T:existence-complex-eigen}
For  $v\in\mathfrak S_{x,y}$ and $\lambda\in\CC^\pm$, there 
exist a unique continuous eigenfunction $\Phi(x,y,\lambda)$ and a positive function $\epsilon(\lambda)$ 
such that 
\begin{gather*}
\Phi=z+\partial_{\bar z}^{-1}\alpha(z,\bar z,\lambda),\ \ z=x-\lambda y,\ \ 
\alpha\in W^{2,p}(dzd\bar z), \ \ \mbox{where} \ \ |p-2|<\epsilon(\lambda)
\end{gather*}

and
\bea
\partial_y\Phi+\left(\lambda+v_x\right)\partial_x\Phi=0,\quad&&\textit{ for $x,\,y\in\mathbb R$, }\label{E:lax-complex}\\
\Phi(x,y,\lambda)-(x-\lambda y )\to 0,\quad && \textit{ as $ x^2+y^2\rightarrow\infty$.} \label{eq:L1_1}
\eea
Moreover, $\Phi(x,y,\cdot)$ is holomorphic for $\lambda\in \mathbb C^\pm$, and 
\beq \label{eq:L1_1-sym}
\Phi(x,y,\lambda)={\overline {\Phi(x,y,\bar\lambda)}}.
\eeq
If $\lambda_I\rightarrow\pm\infty$ we have
\begin{equation}
\label{eq:L1_1-as}
\Phi(x,y,\lambda)=x-\lambda y -\frac{1}{\lambda}v(x,y) + 
o\left( \frac{1}{\lambda} \right).
\end{equation}
\end{theorem}
\begin{proof}
Equation (\ref{E:lax-complex}) takes the following form:
\beq
\label{eq:L1_2}
\left[\partial_{\bar z}+\frac{1}{\lambda-\bar\lambda}v_x(z,\bar z)(\partial_z+\partial_{\bar z}) \right] 
\Phi(z,\bar z,\lambda) = 0, 
\eeq
or equivalently
\beq
\label{eq:L1_3}
\left[ \partial_{\bar z}+ b(z,\bar z,\lambda) \partial_z  \right] \Phi(z,\bar z,\lambda) = 0, 
\eeq
where
\beq
\label{eq:L1_4}
b(z,\bar z,\lambda) =\frac{v_x(z,\bar z)}{2i\lambda_I + v_x(z,\bar z)}.
\eeq
The function $v_x(z,\bar z)$ is real-valued, therefore 
\beq
\label{eq:L1_5}
|b(z,\bar z,\lambda)|  < 1.
\eeq
Using the representation
\[
\partial_z \partial^{-1}_{\bar z} f =\frac1{2\pi i}\iint\frac {f(\zeta,\bar\zeta)}{(\zeta-z)^2}d\zeta \wedge d\bar\zeta,
\]and the Zygmund-Calderon operator theory, it is easy to show that for any fixed $\lambda\in\CC^\pm$
there exist $\varepsilon_2>0$ and $\mu>0$ such that
for $|p-2|<\varepsilon_2$ the norm of the operator 
\beq
\label{eq:L1_7}
f\in W^{2,p}(\mu)\rightarrow b(z,\bar z,\lambda) \partial_z \partial^{-1}_{\bar z} f \in W^{2,p}(\mu)
\eeq
is smaller than 1 \cite{S70}. 
Then we can write \cite{V62}:
\beq
\label{eq:L1_8}
\Phi(z,\bar z,\lambda) = z + \partial^{-1}_{\bar z} \alpha(z,\bar z,\lambda)
\eeq
where $\alpha(z,\bar z,\lambda)$ satisfies the following equation:
\beq
\label{eq:L1_9}
[1+b(z,\bar z,\lambda)\partial_z\partial^{-1}_{\bar z}]\alpha(z,\bar z,\lambda)
+b(z,\bar z,\lambda)=0.
\eeq
This equation is uniquely solvable in the spaces $W^{2,p}$, 
$|p-2|<\varepsilon_2$. Therefore 
$\partial^{-1}_{\bar z} \alpha(z,\bar z,\lambda)$ is a decaying at infinity 
continuous function by Sobolev's theorem, and 
$$
\| \partial^{-1}_{\bar z} \alpha \|_{L^{\infty}(dzd\bar z)}\le 
C_1(\varepsilon_2) \| \alpha \|_ {L^{2-\varepsilon_2}(dz d\bar z)} +
C_2(\varepsilon_2) \| \alpha \|_ {L^{2+\varepsilon_2}(dz d\bar z)}.
$$
We also have:
\beq
\label{eq:L1_10}
\det\left| \begin{array}{cc} \partial_z \Phi &  \partial_{\bar z} \Phi \\ 
 \partial_z \overline{\Phi} &  \partial_{\bar z} \overline{\Phi}
\end{array}
\right| = (1-|b(z,\bar z,\lambda)|^2) | \partial_z \Phi|^2 \ge 0,
\eeq
therefore for all regular points of the map $(z,\bar z)\rightarrow(\Phi,
\overline{\Phi})$ the Jacobian is positive, and the number of preimages is
the same for all regular points. It means that the number of preimages 
is the same for all regular points. This map is one-to-one at infinity, 
therefore it is invertible and we can use $w=\Phi$ as a global coordinate on 
the $z$-plane. In this coordinate all solutions of (\ref{E:lax-complex}) are functions 
holomorphic in $w$ (see Chapter II in \cite{V62}). So Liouville's theorem implies that asymptotics
(\ref{eq:L1_1}) fixes the solution uniquely.

Let us show that $\Phi(x,y,\lambda)$ is holomorphic in $\lambda$
outside the real line. Differentiating (\ref{E:lax-complex}) by $\bar\lambda$ 
we obtain 
\beq
\label{eq:L1_11}
 L \partial_{\bar\lambda} \Phi(x,y,\lambda)=0,
\eeq
and 
\beq
\label{eq:L1_12}
\partial_{\bar\lambda}\Phi(x,y,\lambda) = o(1) \ \ \mbox{as} \ \ x^2+y^2\rightarrow\infty.
\eeq
Therefore $\partial_{\bar\lambda}\Phi(x,y,\lambda)$ is a regular holomorphic 
function in $w$ decaying at infinity, and by 
Liouville's theorem  $\partial_{\bar\lambda}\Phi(x,y,\lambda)\equiv 0$.

The reality condition (\ref{eq:L1_1-sym}) follows from applying   
Liouville's theorem  and the reality conditions $v(x,y)=\overline{v(x,y)}$.

Let $|\lambda_I|\gg 1$. Taking into account, that 
$dz\wedge d\bar z=2i\lambda_I \ dx\wedge dy$ we see, that 
$$
\|\ldots \|_{L^p(dzd\bar z)}= \sqrt[p]{|2\lambda_I|} \cdot \|\ldots\|_{L^p(dxdy)}, \  \ 
\max\limits_{z}|b(z,\bar z, \lambda|= O\left(\frac{1}{\lambda}\right),
$$
$$
\alpha(z,\bar z,\lambda) = - b(z,\bar z,\lambda) + \alpha_1(z,\bar z,\lambda), 
$$
$$
\|\alpha_1(z,\bar z,\lambda) \||_{L^p(dzd\bar z)}\le 
\frac{\max\limits_{z}|b(z,\bar z,\lambda)|}{1-\max\limits_{z}|b(z,\bar z,\lambda)|}\cdot  
\|b(z,\bar z,\lambda) \||_{L^p(dzd\bar z)}=O\left(\frac{\sqrt[p]{|\lambda_I|}}{\lambda_I^2}
\right),
$$
and
$$
\Phi(z,\bar z,\lambda) = z - \frac{1}{2i\lambda_I}\partial^{-1}_{\bar z} v_x(z,\bar z)+ 
o\left( \frac{1}{\lambda_I} \right)= z - \frac{1}{2i\lambda_I}\partial^{-1}_{\bar z} 
(\partial_{\bar z}+\partial_{z} )v(z,\bar z)+ 
o\left( \frac{1}{\lambda_I} \right),
$$
but
$$
\partial_{z} = -\frac{\bar\lambda}{\lambda} \partial_{\bar z}- 
\frac{1}{\lambda} \partial_{y},
$$
therefore
$$
\Phi(z,\bar z,\lambda) = z - \frac{1}{2i\lambda_I}\partial^{-1}_{\bar z} 
\left(\partial_{\bar z}-\frac{\bar\lambda}{\lambda}\partial_{\bar z} \right)
v(z,\bar z)+o\left( \frac{1}{\lambda_I} \right)= 
z - \frac{ v(z,\bar z) }{\lambda} 
+o\left( \frac{1}{\lambda_I} \right). 
$$

\end{proof}

Starting from this point we will work with the Jost eigenfunction $\varphi_{-}$ only; therefore we shall denote it simply by $\varphi$, omitting the subscript:
\beq
\varphi(x,y,\lambda)=\varphi_{-}(x,y,\lambda).
\eeq

\begin{theorem}\label{T:complex-bdry}
Suppose $v\in\mathfrak S_{x,y}$ satisfying (\ref{E:v-supp}). The complex eigenfunction $\Phi(x,y,\lambda)$ has  continuous extensions on $\CC^\pm\cup\RR$. Moreover,  denote the limits on both sides of $\RR$ as $\Phi^\pm$, then $\partial_x^\mu(\Phi^\pm-x+\lambda y)\in W^{1,2}(\mathbb R,d\lambda)$, 
\bea
\Phi^{-}(x,y,\lambda)
&=&\varphi(x,y,\lambda)+\chi_{-}(\varphi(x,y,\lambda),\lambda)\label{psi-phi+}\\
\Phi^+(x,y,\lambda)&=&\overline{\Phi^-(x,y,\lambda)},\label{psi-phi-}
\eea
where 
$\chi_{-}(\xi,\lambda)$ is characterized by the Riemann-Hilbert problem with the shift function $\sigma(\xi,\lambda)$. 
\beq
\label{E:shifted-RH}
\begin{split}
\sigma(\xi,\lambda)+\chi_{+}(\xi+\sigma(\xi,\lambda),\lambda)-\chi_{-}(\xi,\lambda)=0,&\quad \xi\in \RR,\\
\hskip1in\partial_{\bar\xi}\chi=0,&\quad \xi\in\CC^\pm,\\
\chi\to 0,\ |\xi|\to\infty.&\quad \xi\in\CC
\end{split}
\eeq
\end{theorem}
As before, we move the proof to the last Section.

Theorem~\ref{T:complex-bdry} implies 
\beq \label{E:scattering}
\Phi^+(x,y,\lambda)-\Phi^-(x,y,\lambda)=-2i\chi_{-I}(\varphi(x,y,\lambda),\lambda),\ \ \lambda\in\RR,
\eeq
and, due to (\ref{eq:L1_1-as}) 
\begin{equation}
\label{eq:phi-rieamann}
\Phi(x,y,\lambda)= x-\lambda y -\frac{1}{\pi} \int\limits_{\RR} 
\frac{\chi_{-I}(\varphi(x,y,\zeta),\zeta)}{\zeta-\lambda} d\zeta, \ \ \lambda\in\CC^{\pm},
\end{equation}
\begin{equation}
\label{eq:v-rieamann}
v(x,y)= -\frac{1}{\pi} \int\limits_{\RR} \chi_{-I}(\varphi(x,y,\zeta),\zeta) d\zeta.
\end{equation}

\subsection{The shifted Riemann-Hilbert problem}\label{S:S_RH}
In subsection \ref{S:complex}
the following characterization for the boundary value of the complex eigenfunction  
\beq \label{psi-phi+-1}
\Phi^{-}(x,y,\lambda)
=\varphi(x,y,\lambda)+\chi_{-}(\varphi(x,y,\lambda),\lambda),\ \ \textit{for $\lambda\in\RR$},
\eeq
was justified. Here $\chi_{-}(\xi,\lambda)$ satisfies the shifted Riemann-Hilbert problem  (\ref{E:shifted-RH}).

The problem (\ref{E:shifted-RH}) can be converted into the following linear equation \cite{Ga66}
\beq
\label{E:shifted-RH-fredholm}
\chi_{-}(\xi,\lambda)-\frac 1{2\pi i}\int_\RR f(\xi,\xi',\lambda)\chi_{-}(\xi',\lambda)d\xi'+g(\xi,\lambda)=0,
\eeq
where
\beq \label{E:shifted-RH-fredholm-c-2}
\begin{split}
f(\xi,\xi',\lambda)=&\, \frac {\partial_{\xi'}s(\xi',\lambda)}{s(\xi',\lambda)-s(\xi,\lambda)}-\frac 1{\xi'-\xi},\\
g(\xi,\lambda)=&\, -\frac 12\sigma(\xi,\lambda)+\frac 1{2\pi i}\fint_\RR \frac {\partial_{\xi'}s(\xi',\lambda)}{s(\xi',\lambda)-s(\xi,\lambda)}\sigma(\xi',\lambda)d\xi'\\
s(\xi,\lambda)=&\, \xi+\sigma(\xi,\lambda).
\end{split}
\eeq
Under the assumptions that the mapping $\xi\rightarrow\xi+\sigma(\xi,\lambda)$ be invertible for all $\lambda$, that $\sigma(\xi,\lambda)$ decay sufficiently fast for any fixed $\lambda$ and be H$\ddot {\textrm o}$lder continuous, the unique solvability of $\chi_{}$ is proven in \cite{Ga66} by showing a Fredholm alternative  for (\ref{E:shifted-RH-fredholm}). Also this step does not require the small norm assumption.

Our goal in this section is  to obtain some analytic estimates on the spectral data $\chi_{\pm}$, including the large $\lambda$-asymptotic estimates, which are important in characterizing the complex eigenfunction and are indispensable for solving the inverse problem. 

To simplify the calculations we shall use the following agreement in Lemmas~\ref{lem:xi-est-1}-\ref{lem:xi-est-3}: we omit the $\lambda$-dependence in all formulas. It is convenient to denote:
\beq\label{E:K}
\begin{split}
K\psi=&\frac 1{2\pi i}\int_\RR f(\xi,\xi',\lambda)\psi(\xi')d\xi'
\end{split}
\eeq 

It is natural to solve the integral equation (\ref{E:shifted-RH-fredholm}) iteratively. Therefore we have to estimate the norm of $K$, $\partial_t K$.

\begin{lemma}
\label{lem:xi-est-1}
Assume that the scattering data $\sigma(\xi)$, $\xi\in\RR $ satisfy 
the following estimates: 
\begin{enumerate}
\item $\sigma(\xi)$ is 2 times continuously differentiable in $\xi$.
\item $|\sigma(\xi)|\le C_0\le \frac{1}{4}$.
\item $|\sigma'(\xi)|\le C_1\le \frac{1}{2}$.
\item $|\sigma''(\xi)|\le C_2$.
\item $\|\sigma'(\xi)\|_{L^2(d\xi)}\le \hat C_1$.
\end{enumerate}
Then we have the following estimate
$$
\|K\|_{L^{\infty}}\le \frac{1}{\pi}[ 4 C_0+ 2 C_2 + \sqrt{2} \hat C_1].
$$
Assume that, in addition, the scattering data $\sigma(\xi)$, $\xi\in\RR $ satisfy 
the following extra estimates: 
\begin{enumerate}
\item $\sigma(\xi)$ is 3 times continuously differentiable in $\xi$.
\item $C_2\le\frac{1}{2}$.
\item $|\sigma'''(\xi)|\le C_3 $.
\end{enumerate}
Then $K$ maps the space $L^{\infty}(d\xi)$ into the space $C^{1}(\xi)$. 
Moreover, if $h_2(\xi) = (K h_1)(\xi)$, then
$$
|h_{2,\xi}(\xi)|\le \frac{1}{\pi} (2 C_3+ 4 C_2^2 + 14 C_1 +4 C_1^2 ) \cdot
\| h_1(\xi)\|_{L^{\infty}(d\xi)}.
$$
\end{lemma}
The proof of this Lemma is moved to the last Section.

We also require some estimates on the function $g(\xi)$

\begin{lemma}
\label{lem:xi-est-3}
Assume that the scattering data satisfy the same estimates as in Lemma~\ref{lem:xi-est-1} and 
\begin{enumerate}
\item $\|\sigma(\xi)\|_{L^2(d\xi)}\le \hat C_0$.
\item $\|\sigma'(\xi)\|_{L^2(d\xi)}\le \hat C_1$.
\end{enumerate}

Then we have:

\begin{enumerate}
\item $|g(\xi)|\le \frac{C_0}{2} + \frac{1}{\pi} [2 C_1 + \hat C_0  ]$.
\item $|g_{\xi}(\xi)|\le \frac{C_1}{2} + \frac{1}{\pi} [16 C_2 + 4 \hat C_1]$.
\end{enumerate}

Moreover, if $\sigma(\xi)$ has compact support: $\sigma(\xi)=0$ for 
$|\xi|\ge R-1$, then 
\begin{align}
\label{eq:xi-est-4}
|g(\xi)| &\le \frac{C_0}{2} + \frac{6R}{\pi} C_1 \le   \frac{8R}{\pi} C_1 , \\
|g_{\xi}(\xi)| &\le \frac{C_1}{2} + \frac{24R}{\pi} C_2 \le \frac{26R}{\pi} C_2, \nonumber
\end{align}

\end{lemma}
The proof of this Lemma is moved to the last Section.

Combining the estimates from Lemmas~\ref{lem:xi-est-1}-\ref{lem:xi-est-3} we obtain the following:
\begin{proposition}\label{Prop:chi_nprm}
Assume that the potential $v(x,y)$ satisfy the small norm constraints formulated in the 
Definition~\ref{def:small-data}. Then we have
\begin{equation}
\label{eq:chi-main-estimate}
|\chi_{\xi}(\xi,\lambda)|\le \frac{1}{4}\tan\left(\frac{\pi}{8}\right).
\end{equation}
We show below, that this property guaranties the unique solvability of the inverse problem.
\end{proposition}
\begin{proof}
Equation (\ref{E:shifted-RH-fredholm}) can be written in the short form: 
\begin{equation}
\label{E:shifted-RH-fredholm-bis}
\chi_{-} = K \chi_{-} -g,
\end{equation}
If $\|K\|<1$, it can solved iteratively and 
$$
\|\chi_{-}(\xi,\lambda)\| \le \frac{1}{1-\|K\|} \|g\|. 
$$
By Lemma~\ref{lem:xi-est-1}, Condition~3, the small norm conditions list means exactly that
$$
\|K\|_{L^{\infty}(d\xi)}\le\frac{1}{2};
$$
therefore
$$
\|\chi_{-}\|_{L^{\infty}(d\xi)}\le 2 \| g \|_{L^{\infty}(d\xi)} \le
B_0  + \frac{2}{\pi} [2 B_1 + \hat B_0  ].
$$

By differentiating equation (\ref{E:shifted-RH-fredholm-bis}) with respect to 
$\xi$, we obtain:
$$
\chi_{-\xi} = (K \chi_{-})_{\xi} -g_\xi,
$$
and 
$$
\|\chi_{-\xi} \|_{L^{\infty}(d\xi)} \le \| (K \chi_{-})_{\xi} \|_{L^{\infty}(d\xi)} 
+  \|  g_\xi\|_{L^{\infty}(d\xi)}.
$$
By Lemma~\ref{lem:xi-est-1}, in the small norm case
$$
\| (K \chi_{-})_{\xi} \|_{L^{\infty}(d\xi)}\le \frac{1}{\pi} (2 B_3+ 4 B_2^2 + 14 B_1 +4 B_1^2 )\cdot \|\chi_{-} \|_{L^{\infty}(d\xi)}.
$$
By Lemma~\ref{lem:xi-est-3}
$$
 \|  g_\xi\|_{L^{\infty}(d\xi)} \le \frac{B_1}{2} + \frac{1}{\pi} [16 B_2 + 4 \hat B_1].
$$
Therefore 
$$
\|\chi_{-}\|_{L^{\infty}(d\xi)}\le \frac{1}{\pi} (2 B_3+ 4 B_2^2 + 14 B_1 +4 B_1^2 )\cdot
\left(B_0  + \frac{2}{\pi} [2 B_1 + \hat B_0  ]\right) +  \frac{B_1}{2} + \frac{1}{\pi} [16 B_2 + 4 \hat B_1]\le
$$
$$
\le\frac{1}{4}\tan\left(\frac{\pi}{8}\right)<\frac{1}{4}.
$$
\end{proof}

The solution of the inverse problem also requires some estimates on $\chi(\xi,\lambda)$ and its derivatives 
at $\lambda\rightarrow\infty$. Let us show that, at large $\lambda$, the leading term of the asymptotic 
behavior is determined by the linear part of (\ref{E:shifted-RH-fredholm-c-2}). 

More precisely,
\begin{lemma}
\label{lem:xi-est-4}
If $v\in\mathfrak S_{x,y}$ and $v(x,y)=0$ for $|y|\ge D_y$, then, for $\lambda\rightarrow\infty$, we have the following estimates:
\begin{enumerate}
\item $\| K(\lambda) \|_{L^{\infty}(d\xi)} = O \left( \frac{1}{\lambda^2} \right)$.
\item For every sufficiently large  $\lambda$, operator $K(\lambda)$ maps the space $L^{\infty}(d\xi)$ 
into the space $C^{\infty}(\xi)$ and there exists a constant ${\cal C}_1(\lambda)$ such that 
$$
\| (K(\lambda) f )_{\xi} \|_{L^{\infty}(d\xi)} \le {\cal C}_1(\lambda)\cdot \| f \|_{L^{\infty}(d\xi)}, \ \ 
{\cal C}_1(\lambda)=  O \left( \frac{1}{\lambda^2} \right)
$$
\end{enumerate}
\end{lemma}
\begin{proof}
To prove this Lemma it is sufficient to compare formulas (\ref{E:l1-condition-lambda-tau-sigma-0})-(\ref{E:l1-condition-lambda-tau-sigma-1})
with the estimates from Lemma~\ref{lem:xi-est-1}.
\end{proof}
\begin{remark}
Using the same approach, it is possible to prove analogous estimates for all derivatives; in particular,
there exists a constant ${\cal C}_2(\lambda)$ such that 
$$
\| (K(\lambda) f )_{\xi\xi} \|_{L^{\infty}(d\xi)} \le {\cal C}_2(\lambda) \cdot\| f \|_{L^{\infty}(d\xi)}, \ \ 
{\cal C}_2(\lambda)=  O \left( \frac{1}{\lambda^3} \right)
$$
\end{remark}

\begin{proposition}
\label{prop:xi-est-1}
Assume that $v\in\mathfrak S_{x,y}$ and $v(x,y)\equiv0$ for $|y|>D_y$. Then. for 
$\lambda\rightarrow\infty$, we have the following estimates
\begin{align}
\label{eq:xi-est-5}
\chi(\xi,\lambda) &=-g(\xi,\lambda)+O \left( \frac{1}{\lambda^4} \right), & \\
\chi_{\xi}(\xi,\lambda) &=-g_{\xi}(\xi,\lambda)+O \left( \frac{1}{\lambda^4} \right). \nonumber&
\end{align}
If, in addition, $v(x,y)$ satisfies the compact support condition (\ref{E:v-supp}), i.e. $v(x,y)\equiv 0$ for 
$|x|\ge D_x$, then 
\begin{align}
\label{eq:xi-est-6}
\| \chi(\xi,\lambda)\|_{L^{\infty}(d\xi)} &=O \left( \frac{1}{\lambda^2} \right),\\
\| \chi_{\xi}(\xi,\lambda)\|_{L^{\infty}(d\xi)} &= O \left( \frac{1}{\lambda^3} \right).\nonumber
\end{align}
\begin{remark}
Using the same approach, it is possible to prove that in the compact support case
$$
\|\chi_{\xi\xi}(\xi,\lambda)\||_{L^{\infty}(d\xi)} = O \left( \frac{1}{\lambda^4} \right).
$$
\end{remark}
\end{proposition}
Proof of Proposition~\ref{prop:xi-est-1}.
\begin{proof}
From (\ref{E:shifted-RH-fredholm-c-2}), Proposition~\ref{L:direct-sigma-asy-0}, Lemmas~\ref{lem:xi-est-3}, \ref{lem:xi-est-4} and the formula 
$R(\lambda)=D_x+|\lambda| D_y+1$ it follows immediately, that 
$$
\| \chi(\xi,\lambda)+g(\xi,\lambda)\|_{L^{\infty}(d\xi)} \le 
\frac{\| K(\lambda) \|_{L^{\infty}(d\xi)} } {1-\| K(\lambda) \|_{L^{\infty}(d\xi)} } \cdot 
\| g(\xi,\lambda)\|_{L^{\infty}(d\xi)} = O \left( \frac{1}{\lambda^3} \right),
$$
$$
\|\chi_{\xi}(\xi,\lambda)+g_{\xi}(\xi,\lambda)\|_{L^{\infty}(d\xi)} \le {\cal C}_1(\lambda)  \| \chi(\xi,\lambda) \|_{L^{\infty}(d\xi)} = O \left( \frac{1}{\lambda^4} \right).
$$
\end{proof}

\begin{proposition}\label{P:decay}
Suppose $v\in\mathfrak S_{x,y}$ with compact support and $v$ is small. Consider a curve in the $(\xi,\lambda)$-plane:
$$
\xi(\lambda) = x-\lambda y-\lambda^2 t+ \omega(\lambda).
$$
Then for fixed $t>0$ and $\omega(\lambda)=O(1)$, we have, as $\lambda\to\infty$,
\beq\label{E:decay}
|\partial^\mu_\xi\chi_{-}(\omega(\lambda)+x-\lambda y-\lambda^2 t,\lambda)|=  
\mathcal O\left(\frac{1}{1+|\lambda|^{3+2\mu}}\right).
\eeq
\end{proposition}
\begin{proof}
Suppose the support of $v$ is contained in $\{|x|\le D_x,\ |y|\le D_y\}$. Therefore the support of 
$\sigma(\xi,\lambda)$ lies in the area  $|\xi|\le |D_x+|\lambda|D_y|$, $\xi\in\RR$. 

\begin{figure}[h]
\label{fig2}
\begin{center}
\includegraphics[height=6cm]{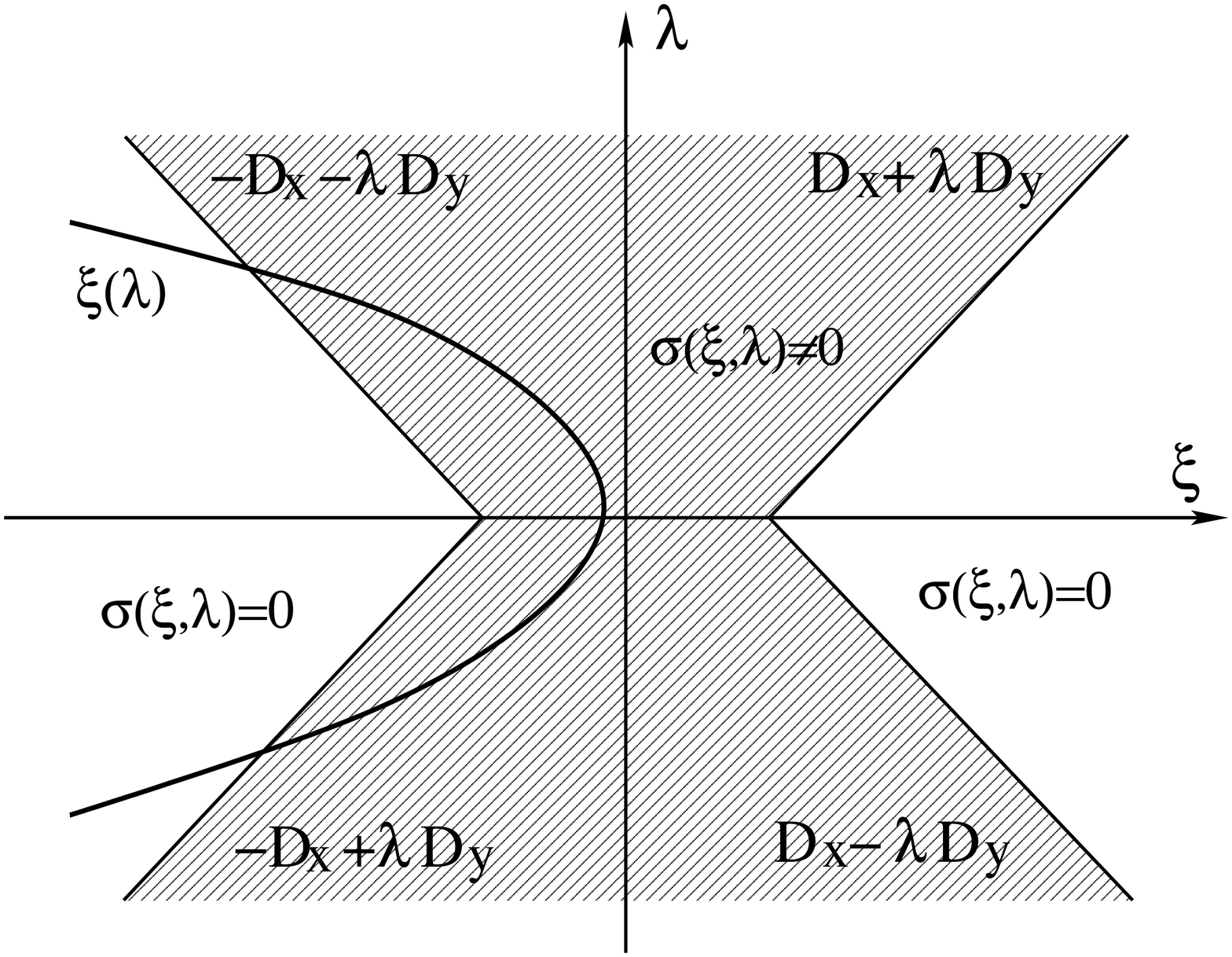}
\end{center}
\caption{The support of $\chi(\xi,\lambda)$ and the trajectory $\xi(\lambda)$, $t>0$.}
\end{figure}

Outside this area  $\sigma(\xi,\lambda))\equiv 0$, hence $\chi(\xi,\lambda)$
is holomorphic in $\xi$ in the complex plane outside the cut $[-D_x-|\lambda|D_y,D_x+|\lambda|D_y ]$ on the real line. Therefore, 
\[
\chi(\xi,\lambda) = 
\frac{1}{2\pi i}\int_{-D_x-|\lambda|D_y}^{D_x+|\lambda|D_y}\frac{(\chi_{+}(\tau,\lambda)-\chi_{-}(\tau,\lambda))}{\tau-\xi}d\tau,\]
and
\[
\partial_\xi\chi(\xi,\lambda) = 
\frac{1}{2\pi i}\int_{-D_x-|\lambda|D_y}^{D_x+|\lambda|D_y}\frac{(\chi_{+}(\tau,\lambda)-\chi_{-}(\tau,\lambda))}{(\tau-\xi)^2}d\tau.
\]
It follows that, for $|\xi|> D_x+|\lambda|D_y $,
\begin{eqnarray*}
&&|\partial^{\mu}_\xi\chi_{-}(\xi,\lambda) | \\
&\le & \frac{1}{\pi}\cdot [\|\chi_{-}(\xi,\lambda)\|_{L^{\infty}} + |\chi_{+}(\xi,\lambda) \|_{L^{\infty}}] \cdot (D_x+|\lambda|D_y) \cdot
\sup_{\tau\in[-D_x-|\lambda|D_y,+D_x+|\lambda|D_y]}\partial^{\mu}_{\xi}\left[\frac{1}{(\tau-\xi)}\right]\\
&\le &  C\frac{\mu !}{(1+|\lambda|)[|\xi|-(D_x+|\lambda|D_y)]^{\mu+1}}
\end{eqnarray*}
Therefore (\ref{E:decay}) follows if $t>0$.
\end{proof}

\section{The inverse problem}\label{S:NRH}
\subsection{The reconstruction of the real eigenfunction }\label{S:NRH-1}

Assume that the {\bf spectral data} $\chi_{-}(\xi,\lambda)=\chi_{-R}(\xi,\lambda)+i\chi_{-I}(\xi,\lambda)$ are given, where $\xi,\lambda\in\RR$. Let us recall that $\chi_{-}(\xi,\lambda)$ is assumed to be analytic in $\xi$ the lower half-plane, or equivalently
\beq
\label{disp-realtions}
\chi_{-R}-H_{\xi}\chi_{-I}=0,
\eeq
where $H_{\xi}$ denotes the Hilbert transform wrt. the variable $\xi$.

Our current aim is to construct the common eigenfunctions of the Lax pair for the Pavlov equation starting from the spectral data. By the Plemelj (Sokhotski) formula \cite{Ga66}, Theorem \ref{T:existence-complex-eigen}, and \ref{T:complex-bdry}, we have
\begin{equation}\label{E:Manakov-Santini-0}
\varphi(x,y,\lambda)+\chi_{-R}(\varphi(x,y,\lambda),\lambda)=x-\lambda y -\frac 1{\pi}\fint_{\mathbb R}\frac{\chi_{-I}(\varphi(x,y,\zeta),\zeta)}{\zeta-\lambda}d\zeta.
\end{equation}

Therefore, keeping in mind the time evolution (\ref{eq:spectral-evolution}) of the spectral data and the definition (\ref{phi_Phi}) of the common eigenfunctions of the vector field Lax Pair, the nonlinear integral equation of the inverse problem reads:
\begin{equation}\label{E:Manakov-Santini}
\psi(x,y,t,\lambda)+\chi_{-R}(\psi(x,y,t,\lambda),\lambda)=x-\lambda y-\lambda^2 t -\frac 1{\pi}\fint_{\mathbb R}\frac{\chi_{-I}(\psi(x,y,t,\zeta),\zeta)}{\zeta-\lambda}d\zeta 
\end{equation}

The solution of the inverse problem consists of two steps:
\begin{itemize}
\item We show that, if some appropriate constraints are imposed on the spectral data, the nonlinear integral equation (\ref{E:Manakov-Santini}) has a unique solution $\psi(x,y,t,\lambda)$.
\item We show that the function $\psi(x,y,t,\lambda)$ is the real Jost eigenfunction for the Pavlov Lax Pair wth the proper behavior at $y\rightarrow-\infty$, where the potential $v(x,y,t)$ is defined by formula (\ref{E:potential-inv-S}).
\end{itemize}

\begin{theorem} \label{T:uniqueness-S2} {\bf (Global solvability for the IST equation~(\ref{E:Manakov-Santini}) -- part 1.)} 
Suppose that the spectral data $\chi_{-}(\xi,\lambda)$ satisfy the following constraints
\begin{enumerate}
\item $\chi_{-}(\xi,\lambda)$, $\partial_{\xi}\chi_{-}(\xi,\lambda)$ are well-defined continuous functions.
\item 
$$
|\partial_{\xi}\chi_{-R}(\xi,\lambda)| \le \frac{1}{4}\tan\left(\frac{\pi}{8}\right), \ \ 
|\partial_{\xi}\chi_{-I}(\xi,\lambda)| \le \frac{1}{4}\tan\left(\frac{\pi}{8}\right).
$$
\item There exists a positive constant $C$ such that 
$$
|\chi_{-}(\xi,\lambda)|\le \frac{C}{1+|\lambda|}  
$$
\end{enumerate}
Then, for all $x,y,t\in\RR$, $t\ge 0$, equation (\ref{E:Manakov-Santini}) has a unique solution $\psi(x,y,t,\lambda)$ such that $\psi(x,y,t,\lambda)=x-\lambda y-\lambda^2 t+\omega(x,y,t,\lambda) $, where $\omega(x,y,t,\lambda)\in L^2(d\lambda)$.  
\end{theorem}
\begin{proof}
The proof is based on the standard iteration procedure for contracting nonlinear maps. Equation (\ref{E:Manakov-Santini}) is equivalent to 
\beq
\omega(x,y,t,\lambda)= {\mathcal R}\left[\omega(x,y,t,\lambda)+x-\lambda y-\lambda^2 t\right],\label{E:FP-ite}
\eeq
where $\mathcal R$ is defined by
\begin{equation}\label{E:spectral-operator}
({\mathcal R}[f])(\lambda)=- \chi_{-R}(f(\lambda),\lambda)-\frac{1}{\pi}\fint_{\RR}\frac{\chi_{-}(f(\zeta)),\zeta)}{\zeta-\lambda}d\zeta
\end{equation}
or equivalently,
$$
({\mathcal R}[f])(\lambda)=H_{\lambda}\circ \chi_{-I}( f(\lambda),\lambda)- \chi_{-R}(f(\lambda),\lambda).
$$
From the constraints on the spectral data it immediately follows that the maps 
$$ 
f(\lambda) \rightarrow \chi_{-R}( f(\lambda),\lambda), \ \  f(\lambda) \rightarrow \chi_{-I}( f(\lambda),\lambda)
$$
map all measurable functions of $\lambda$ into the space $L^2(d\lambda)$; moreover the image of the map is located inside the ball of radius $R_0=
\sqrt{2}C$.
$H_{\lambda}$ is a unitary operator in the space $L^2(\lambda)$; therefore, for any measurable function $f(\lambda)$, we know that  
${\mathcal R}[f]\in L^2(d\lambda)$, and $\| {\mathcal R} [f] \|_{L^2(d\lambda)}\le 
2\sqrt{2}C$.

Let us check that operator $\mathcal R$ is a contraction. Let $f(\lambda)$ be a measurable function, $g(\lambda)\in L^2(d\lambda)$. We have
$$
\|\mathcal R[f+g]-\mathcal R[f]\|_{L^2(d\lambda)}\le
$$
$$ 
\le\|H_{\lambda}\circ [\chi_{-I}( f(\lambda)+g(\lambda),\lambda)- \chi_{-I}( f(\lambda),\lambda)] \|_{L^2(d\lambda)}+ 
\|\chi_{-R}( f(\lambda)+g(\lambda),\lambda)- \chi_{-R}( f(\lambda),\lambda) \|_{L^2(d\lambda)}=
$$
$$
=\| \chi_{-I}( f(\lambda)+g(\lambda),\lambda)- \chi_{-I}( f(\lambda),\lambda) \|_{L^2(d\lambda)}+ 
\|\chi_{-R}( f(\lambda)+g(\lambda),\lambda)- \chi_{-R}( f(\lambda),\lambda) \|_{L^2(d\lambda)}
$$
We know, that 
$$
|\chi_{-I}( f(\lambda)+g(\lambda),\lambda)- \chi_{-I}( f(\lambda),\lambda)|\le\max\limits_{\xi,\lambda} |\partial_{\xi}\chi_{-I}(\xi,\lambda)|\cdot|g(\lambda)|
\le\frac{1}{4}|g(\lambda)|,
$$
$$
|\chi_{-R}( f(\lambda)+g(\lambda),\lambda)- \chi_{-R}( f(\lambda),\lambda)|\le\max\limits_{\xi,\lambda} |\partial_{\xi}\chi_{-R}(\xi,\lambda)|\cdot|g(\lambda)|
\le\frac{1}{4}|g(\lambda)|,
$$
therefore
$$
\|{\mathcal R}[f+g]-{\mathcal R}[f]\|_{L^2(d\lambda)}\le\frac{1}{2} \| g\|_{L^2(d\lambda)}.
$$
Hence the iteration procedure:
\bea
&&\omega_0(x,y,t,\lambda)=0 \\
&&\omega_{n+1}(x,y,t,\lambda)= {\mathcal R}[\omega_{n}(x,y,t,\lambda)+x-\lambda y-\lambda^2 t],\label{E:FP-ite2}
\eea
perfectly converges in $L^2(d\lambda)$.
\end{proof}
Let us check now that the functions constructed above have the Jost property. Namely:

\begin{theorem} \label{T:JOST}
Assume that the spectral data $\chi_{-}(\xi,\lambda)$ satisfy the same constraints as in Theorem~\ref{T:uniqueness-S2}, and
\begin{enumerate}
\item For each $\lambda\in\RR$ the function $\chi(\xi,\lambda)$ is holomorphic in $\xi$ in the lower half-plane.
\item $\partial_{\lambda}\chi_{-}(\xi,\lambda)$, 
$\partial_{\xi}\partial_{\lambda}\chi_{-}(\xi,\lambda)$,  
$\partial^2_{\xi}\chi_{-}(\xi,\lambda)$
are well-defined continuous functions.
\item There exists a positive constant $C$ such, that 
\begin{align}
&|\chi_{-}(\xi,\lambda)|\le \frac{C}{1+|\lambda|^2}, \label{T:JOST:eq1}\\
&|\chi_{-}(\xi,\lambda)|\le \frac{C}{1+|\xi|}, \label{T:JOST:eq2}\\
&|\partial_{\xi}\chi_{-}(\xi,\lambda)|\le \frac{C}{1+|\lambda|^3},  \label{T:JOST:eq3}\\
&|\partial_{\xi}\chi_{-}(\xi,\lambda)|\le \frac{C}{1+|\xi|^2},  \label{T:JOST:eq4}\\
&|\partial^2_{\xi}\chi_{-}(\xi,\lambda)|\le C,   \label{T:JOST:eq5}\\
&|\partial_{\lambda}\chi_{-}(\xi,\lambda)|\le \frac{C}{1+|\lambda|}.  \label{T:JOST:eq6}
\end{align}
\item For any $\calD>0$, there exists a positive constant $C(\calD)$) such that, for all 
$\lambda$ such that $|\lambda|\le 4\calD$
\begin{align}
&|\partial_{\lambda}\chi_{-}(\xi,\lambda)|\le \frac{C(\calD)}{1+|\xi|},  \label{T:JOST:eq7}\\
&|\partial_{\xi}\chi_{-}(\xi,\lambda)|\le \frac{C(\calD)}{1+|\xi|^2},  \label{T:JOST:eq8}\\
&|\partial_{\lambda}\partial_{\xi}\chi_{-}(\xi,\lambda)|\le \frac{C(\calD)}{1+|\xi|^2},  
\label{T:JOST:eq9}\\
&|\partial^2_{\xi}\chi_{-}(\xi,\lambda)|\le \frac{C(\calD)}{1+|\xi|^3}.  \label{T:JOST:eq10}
\end{align}
\end{enumerate}
Then, for the functions constructed in Theorem~\ref{T:uniqueness-S2} with fixed $\tau,t,\lambda_0\in\RR$, $t\ge0$, we have
\beq
\label{E:Jost}
\omega(\tau+\lambda_0 y,y,t,\lambda_0) \rightarrow 0  \ \ \mbox{for} \ \ y\rightarrow-\infty.
\eeq
\end{theorem}

\begin{remark} 
Let us point out that all conditions from Theorem~\ref{T:JOST} holds for the spectral data constructed in the framework of the direct spectral transform 
(we assume again that our Cauchy data $v_0(x,y)$ have compact support). Almost all of them were proved above, and the proof of the remaining ones are rather standard.  Let us check, for example, (\ref{T:JOST:eq2})

Let $|\xi|$ be sufficiently large. We have 2 regions. 
\begin{enumerate}
\item Let $|\lambda|^2\ge |\xi|$. Then the second condition immediately follows from  (\ref{eq:xi-est-6})
\item Let $|\lambda|^2< |\xi|$. From Proposition \ref{Prop:chi_nprm} and  (\ref{eq:xi-est-6}) we obtain that there exists a constant $C_0$ such that 
\beq
\int\limits_{-\infty}^{\infty} |\chi_{+}(\tau,\lambda)-\chi_{-}(\tau,\lambda)|
d\tau < C_0 \ \ \mbox{for all} \ \ \lambda.
\eeq
If $|\xi|>2 |D_x+ \sqrt{|\xi|}D_y |$, then we can use the same estimates as in 
Proposition~\ref{prop:xi-est-1}, and 
\beq
|\chi(\xi,\lambda)|\le \frac{C_0}{\pi|\xi|} 
\eeq
It completes the proof.
\end{enumerate}

\end{remark}

\begin{remark} 
One can consider equation (\ref{E:Manakov-Santini}) without assuming that the  
spectral data $\chi_{-}(\xi,\lambda)$ is holomorphic in $\xi$ in the lower half-plane 
(or, equivalently, we do not assume that equation   (\ref{disp-realtions}) is fulfilled). In this situation the 
function $\psi(x,y,t,\lambda)$ will be also an eigenfunction for the Pavlov Lax operators 
$L$, $M$ for some $v(x,y,t)$ (see Theorem~\ref{T:uniqueness-S}), but the normalization of this eigenfunction will be different from(\ref{E:Jost}). 
\end{remark}

We also require to study the linearized version of  equation~(\ref{E:Manakov-Santini}).

\begin{lemma}\label{L:gn-inv-S}  Suppose that the scattering data $\chi_{-}(\xi,\lambda)$ satisfy the same constraints as in Theorem~\ref{T:uniqueness-S2} (which are fulfilled if $\chi_{-}(\xi,\lambda)$ was constructed through the small norm Cauchy data $v(x,y)$). Then, for $\forall g\in L^p(\RR,d\lambda)$, $p=2,4$ the integral equation
\beq\label{E:linearize-RH-S}
\begin{split}
f(x,y,t,\lambda)=&\frac 1{\pi}\fint_\RR \frac {\partial_\xi\chi_{-I}(\psi_{-}(x,y,t,\zeta) ,\zeta)}{\lambda-\zeta}f(x,y,t,\zeta)\,d\zeta\\
& -\partial_\xi\chi_{-R}(\psi_{-}(x,y,t,\lambda),\lambda)f(x,y,t,\lambda)+g(\lambda).
\end{split}
\eeq
admits a unique solution $f$ such that
\begin{equation}\label{E:gn-inv-1-S} 
\|f(x,y,t,\lambda)\|_{L^p(\RR,d\lambda)}  
\le 2\|g\|_{L^p(\RR,d\lambda)} .
\end{equation}
\end{lemma}

\begin{proof}
The Hilbert transform is a unitary operator in $L^2(\RR,d\lambda)$ and the norm of the Hilbert transform in 
$L^p(\RR,d\lambda)$, $2\le p<\infty$ is equal to $\cot\left(\frac{\pi}{2p}\right)$ (see \cite{GK,Pic}); therefore for $p=2$ or $p=4$ one has
\beq\label{E:small-data-norm}
\begin{split}
& \left\|\left.\left(\partial{\mathcal R}/{\partial \xi}\right) \right|_{\psi_{-}}(f)\right\|_{L^p(\RR,d\lambda)} \\
\le &\ 2\cot\left(\frac{\pi}{8}\right)\|\partial_\xi\mathcal \chi(\xi,\lambda) \ f\|_{L^p(\RR,d\lambda)}\\
\le &\ 2\cot\left(\frac{\pi}{8}\right)\|\partial_\xi\mathcal \chi(\xi,\lambda)\|_{L^\infty}\|f\|_{L^p(\RR,d\lambda)} \\
\le &\ \frac{1}{2}\|f\|_{L^p(\RR,d\lambda)}.
\end{split}
\eeq
Therefore $[1-\left(\partial{\mathcal R}/{\partial \xi}\right)|_{\psi_{-}}]$ is an invertible map  on $L^p(\RR,d\lambda)$ and the norm of the inverse operator in $L^p(\RR,d\lambda)$ is not greater than 2:
\[
f(x,y,t,\lambda)=\left(1-\partial{\mathcal R}/{\partial \xi}|_{\psi_{-}}\right)^{-1}g\quad \in L^p(\RR,d\lambda)
\]
and the estimate (\ref{E:gn-inv-1-S}) follows.
\end{proof}

Below we use the following simple corollary of the Sobolev embedding theorem:
\begin{lemma}\label{Sobolev-embedding-1}
Let $f(\lambda)$ be an element of $H^1(d\lambda)$, $\lambda\in\RR$. Then $f(\lambda)$ is a 
continuous function and 
\begin{equation}
|f(\lambda)|\le \sqrt{\|f\|_{L^2(d\lambda)}\cdot\|f_{\lambda}\|_{L^2(d\lambda)}}
\end{equation}
\end{lemma}

\begin{theorem} \label{T:uniqueness-S} {\bf (Global solvability for the IST equation~(\ref{E:Manakov-Santini}) -- part 2.)} 
Suppose that $\chi(\xi,\lambda)$ satisfies the same constraints as in Theorem~\ref{T:uniqueness-S2} and, in addition,
\begin{align}
\label{eq:xi-est-6-bis}
\| \partial_{\xi}^n\chi(\xi,\lambda)\|_{L^{\infty}(d\xi)} &=O \left( \frac{1}{\lambda^{2+n}} \right), \ \ n=0,1,2,3,   \nonumber \\
\| \partial_{\xi}^n\partial_{\lambda}\chi(\xi,\lambda)\|_{L^{\infty}(d\xi)} &=O \left( \frac{1}{\lambda^{3+n}} \right), \ \ n=0,1.
\end{align}
Let us denote:
$$
\omega = \omega(x,y,t,\lambda)=\psi(x,y,t,\lambda)-( x-\lambda y-\lambda^2 t)
$$
Then:
\begin{enumerate}
\item For all $x,y\in\RR$, $t\ge0$ the function $\omega(x,y,t,\lambda)$ lies 
in the space $H^1(d\lambda)$ and continuously depends on $x,y,t$ as an element 
of $L^2(d\lambda)\cap L^{\infty}(d\lambda)$. The norm of $\omega$ in the space
$L^2(d\lambda)$ is uniformly bounded in $x,y,t$ (but the $H^1$-norm may 
be unbounded).
\item
For all $x,y\in\RR$, $t\ge 0$ the following derivatives of $\omega$: 
$$
\partial_x \omega, \ \ \partial_y \omega, \ \ \partial_t \omega,  \ \ \partial^2_x \omega, 
\ \ \partial^2_y \omega, \ \ \partial_x\partial_y \omega, \ \ \partial_t\partial_x \omega,
$$
are well-defined as elements of the space $L^2(\RR,d\lambda)$, and $\psi(x,y,0,\lambda)=\varphi(x,y,\lambda)$,
continuously depend on $x,y,t$ and are uniformly bounded in $\RR\times\RR\times \overline{\RR^+}$.
\end{enumerate}
\end{theorem}
\begin{proof} 
\begin{enumerate}
\item
To construct $\omega$, it is convenient to run the iteration procedure  
(\ref{E:Manakov-Santini}), simultaneously for $\omega$ and $\omega_{\lambda}$:
\begin{equation}
\label{E:Manakov-Santini-2}
\begin{split}
\omega^{(n+1)}=&-\chi_{-R}(x-\lambda y -\lambda^2 t + \omega^{(n)},\lambda)+
H_{\lambda}\left[ \chi_{-I}(x-\lambda y -\lambda^2 t + \omega^{(n)},\lambda)  
\right]\\
\omega_{\lambda}^{(n+1)}=&g^{(n)}_{\lambda} -\partial_{\xi}\chi_{-R}(x-\lambda y -\lambda^2 t + \omega^{(n)},\lambda)\,\omega_{\lambda}^{(n)}+\\
&+H_{\lambda}\left[ \partial_{\xi}\chi_{-I}(x-\lambda y -\lambda^2 t + \omega^{(n)},\lambda)\, \omega_{\lambda}^{(n)} \right],
\end{split}
\end{equation}
where
$H_{\lambda}$ is the Hilbert transform with respect to $\lambda$,
\beq
\begin{split}
g^{(n)}_{\lambda}=&-\partial_{\lambda}\chi_{-R}(x-\lambda y -\lambda^2 t + \omega^{(n)},\lambda)+\\
&+H_{\lambda}\left[ \partial_{\lambda}\chi_{-I}(x-\lambda y -\lambda^2 t + \omega^{(n)},\lambda) \right]+ \\
&+(\partial_{\xi}\chi_{-R}(x-\lambda y -\lambda^2 t + \omega^{(n)},\lambda))\cdot
(y+2\lambda t)-\\
&-H_{\lambda}\left[ (\partial_{\xi}\chi_{-I}(x-\lambda y -\lambda^2 t + \omega^{(n)},\lambda))\cdot(y+2\lambda t) \right]
\end{split}
\eeq
In any compact area in the $x,y,t$ space the function $g^{(n)}_{\lambda}$ is bounded in 
$L^2(d\lambda)$ uniformly in $\omega$. If $\|g^{(n)}_{\lambda} \|_{L^2(d\lambda)}<F$,
then for all $n$,  $\|\omega_{\lambda}^{(n)} \|_{L^2(d\lambda)}<2F$. 
Therefore by Lemma~\ref{Sobolev-embedding-1} the $L^2(d\lambda)$ convergence of 
$\omega^{(n)}$ implies the $L^{\infty}(d\lambda)$ convergence of $\omega^{(n)}$ and the convergence of  $\omega_{\lambda}^{(n)}$ in $L^2(d\lambda)$.

\item
By taking derivatives of both sides of (\ref{E:Manakov-Santini}), we  obtain the linearized integral equation by:
\bea
&\psi_x+\partial_{\xi}\chi_{-R}(\psi,\lambda) \cdot\psi_x=
1+H_{\lambda}\left[\partial_{\xi}\chi_{-I}(\psi,\lambda)\psi_x(x,y,t,\lambda)
\right],\label{E:Manakov-Santini-d-x-S}\\
&\psi_y+ \partial_{\xi}\chi_{-R}(\psi,\lambda) \cdot\psi_y=
-\lambda-H_{\lambda}\left[\partial_{\xi}\chi_{-I}(\psi,\lambda)\psi_y(x,y,t,\lambda)
\right],\label{E:Manakov-Santini-d-y-S}\\
&\psi_t+ \partial_{\xi}\chi_{-R}(\psi,\lambda) \cdot \psi_t=
-\lambda^2+H_{\lambda}\left[\partial_{\xi}\chi_{-I}(\psi,\lambda)\psi_t(x,y,t,\lambda)
\right].\label{E:Manakov-Santini-d-t-S}\
\eea
In terms of $\omega(x,y,t,\lambda)$, equations (\ref{E:Manakov-Santini-d-x-S})-(\ref{E:Manakov-Santini-d-t-S}) take the form:
\begin{equation}\label{E:Manakov-Santini-d-y-r-S}
\omega_{\alpha}(x,y,t,\lambda)= g_{\alpha}+H_{\lambda}[\partial_{\xi}\chi_{-I}(\psi,\lambda)\,
\omega_{\alpha}]-\partial_{\xi}\chi_{-R}(\psi,\lambda)\,\omega_{\alpha}
\end{equation}
where $\alpha\in\{x,y,t\}$, and
\beq\label{E:g-1-2-S}
\begin{split}
g_x(x,y,t,\lambda)=&H_{\lambda}[\partial_{\xi}\chi_{-I}(\psi,\lambda)]-
\partial_{\xi}\chi_{-R}(\psi,\lambda),\\
g_y(x,y,t,\lambda)=&-H_{\lambda}[\partial_{\xi}\chi_{-I}(\psi,\lambda)\,\lambda]+
\partial_{\xi}\chi_{-R}(\psi,\lambda)\,\lambda,\\
g_t(x,y,t,\lambda)=& -H_{\lambda}[\partial_{\xi}\chi_{-I}(\psi,\lambda)\,\lambda^2]+
\partial_{\xi}\chi_{-R}(\psi,\lambda)\,\lambda^2.
\end{split}
\eeq
From (\ref{eq:xi-est-6-bis}) it follows that $g_x,\,g_y,\,g_t\in L^2(\RR,d\lambda)\cap  L^4(\RR,d\lambda)$.
Therefore the existence of $\psi_x$, $\psi_y$, $\psi_t$ such that $\partial_x(\psi-( x-\lambda y-\lambda^2 t))$, $\partial_y(\psi-( x-\lambda y-\lambda^2 t))$, $\partial_t(\psi-( x-\lambda y-\lambda^2 t))\in L^2(\RR,d\lambda)\cap  L^4(\RR,d\lambda)$ follows from Lemma \ref{L:gn-inv-S}. 

For the second derivatives of the wave function we have:
\bea
&\psi_{\alpha\beta}=-\partial_{\xi}^2\chi_{-R}(\psi,\lambda) \, \psi_\alpha\psi_\beta-
\partial_{\xi}\chi_{-R}(\psi,\lambda)\, \psi_{\alpha\beta}+\nonumber \\
&+H_{\lambda}\left[\partial_{\xi}^2\chi_{-I}(\psi,\lambda)\, \psi_\alpha\psi_\beta\right]+
H_{\lambda}\left[\partial_{\xi}\chi_{-I}(\psi,\lambda)\, \psi_{\alpha\beta}\right]=
\label{E:Manakov-Santini-d-xx-S}\\
&=g_{\alpha\beta}-\partial_{\xi}\chi_{-R}(\psi,\lambda)\,\psi_{\alpha\beta}+
H_{\lambda}\left[\partial_{\xi}\chi_{-I}(\psi,\lambda)\,\psi_{\alpha\beta}\right],\nonumber
\eea
where
$$
g_{\alpha\beta}=-\partial_{\xi}^2\chi_{-R}(\psi,\lambda)\,\psi_\alpha\psi_\beta+
H_{\lambda}\left[ \partial_{\xi}^2\chi_{-I}(\psi,\lambda)\,\psi_\alpha\psi_\beta\right].
$$

From (\ref{eq:xi-est-6-bis}) and the properties of the first derivatives we obtain that $g_{xx}$,  $g_{xy}$,
$g_{xt}$, $g_{xy}$ belong to  $L^2(\RR,d\lambda)$; therefore equations~(\ref{E:Manakov-Santini-d-xx-S}) are 
uniquely solvable in   $L^2(\RR,d\lambda)$. 

Taking into account that $\omega$ is continuous in $x,y,t$ as an element of 
$L^{\infty}(d\lambda)$, we obtain that all coefficients of the linear equations are continuous 
in $L^{\infty}(d\lambda)$. This implies that the solutions are also continuous.
\end{enumerate}
\end{proof}

\subsection{Eigenfunctions of the Lax equation and the Cauchy problem}\label{S:lax-cauchy}

\begin{theorem}\label{L:eigenfunction-inv-0-S} {\bf (Global solvability for small initial data)} 
Suppose $v_0(x,y)\in\mathfrak S_{x,y}$ satisfying (\ref{E:v-supp}) and the sufficiently small condition from Definition~\ref{def:small-data}.
Let $\psi(x,y,t,\lambda)$ be the solution of the nonlinear inverse problem (\ref{E:Manakov-Santini}) obtained in Theorem \ref{T:uniqueness-S} with the data $\chi(\xi,\lambda)$ constructed 
from $v_0(x,y)$ through the direct problem. Define
\begin{equation}
v(x,y,t)=-\frac 1{\pi}\int_{\mathbb R}\chi_{-I}(\psi(x,y,t,\lambda),\lambda)d\lambda,\label{E:potential-inv-S}
\end{equation}
Then
\begin{enumerate}
\item
\begin{gather}
v(x,y,t)=\overline {v(x,y,t)}, \label{E:potential-regularity-S}\\ 
v,\,v_x,\,v_y,\,v_{xx},\,v_{xy},v_{xt},\,v_{yy} \in 
C(\RR\times\RR\times \overline{\RR^+})\cap L^\infty(\RR\times\RR\times \overline{\RR^+}),
\label{E:eigenfunction-inv-1-S}
\end{gather}
\item Assume, in addition, that for $\chi(\xi,\lambda)$ we have estimates from Proposition~\ref{P:decay}. Then,
for all $t>0$, the function $v_{t}(x,y,t)$ is well-defined and continuous in all variables. 
\item  Function $\psi(x,y,t,\lambda)$ satisfies the Lax equations in the 
space $L^2(d\lambda)$. More precisely, for each $x,y,t$, functions 
$L\psi$, $M\psi$ are well-defined elements of $L^2(d\lambda)$ and 
\begin{eqnarray} 
&&L\psi=\partial_y\psi+\left(\lambda+v_x\right)\partial_x \psi=0, 
\label{E:pavlov-inv-x-S}\\
&&M\psi=\partial_t\psi+\left(\lambda^2+\lambda v_x-v_y\right)\partial_x\psi=0 \label{E:pavlov-inv-y-S}
\end{eqnarray}
for almost all $\lambda\in\RR$.
\item
Let us define a pair of functions  $\Psi^{\pm}(x,y,t,\lambda)$, $\lambda\in\RR$ by
\begin{align}
\label{E:analytic-1}
&\Psi^{-}(x,y,t,\lambda) = \psi(x,y,t,\lambda) + 
\chi_{-}(\psi(x,y,t,\lambda),\lambda),
\nonumber
\\
&\Psi^{+}(x,y,t,\lambda+i0) = \psi(x,y,t,\lambda) + 
\overline{\chi_{-}}(\psi(x,y,t,\lambda),\lambda).
\end{align}
Then, for each $x,y\in\RR$, $t\ge0$, these functions admit natural analytic 
continuation in $\lambda$ to the lower half-plane $\CC^-$ and the upper half-plane 
$\CC^+$ respectively.
\item Denote by $\Psi(x,y,t,\lambda)$, $\lambda\in\CC\backslash\RR$ the function, 
coinciding  with the analytic 
continuation of $\Psi^{+}(x,y,t,\lambda)$ for $\Im\lambda>0$ and with the analytic 
continuation of  $\Psi^{-}(x,y,t,\lambda)$ for $\Im\lambda<0$. Then 
we have the following integral representation:
\begin{equation}
\Psi(x,y,t,\lambda)=x-\lambda y-\lambda^2 t-\frac 1{\pi}\int_{\mathbb R}\frac {\chi_{-I}(\psi(x,y,t,\zeta),\zeta)}{\zeta-\lambda}d\zeta, \ \lambda\in\mathbb C ^\pm.
\label{E:eigenfunction-inv-S}
\end{equation}
Denote by $\hat\omega=\hat\omega(x,y,t,\lambda)$ the regular part of the  wave function: 
$\hat\omega=[\Psi-(x-\lambda y-\lambda^2 t)]$. 

Then for each fixed $\lambda\in\CC^\pm\backslash\RR$ we have:
\begin{equation}\label{E:eigen-condition-x-S}
\hat\omega,\,\hat\omega_x,\,\hat\omega_y,\,\hat\omega_t,\,\hat\omega_{xx},\,\hat\omega_{xy},
\hat\omega_{xt},\,\hat\omega_{yy} \in 
C(\RR\times\RR\times \overline{\RR^+})\cap L^\infty(\RR\times\RR\times \overline{\RR^+}), 
\end{equation}
and for any $\lambda\in\CC^\pm\backslash\RR$ the analytic wave function $\Psi(x,y,t,\lambda)$ satisfies the Lax pair
\begin{eqnarray}
&&L\Psi=\partial_y\Psi+\left(\lambda+v_x\right)\partial_x \Psi=0,\\
&&M\Psi=\partial_t\Psi+\left(\lambda^2+\lambda v_x-v_y\right)\partial_x\Psi=0,
\end{eqnarray}
identically in $x,y,t$.

\item For $t=0$ the function $v(x,y,t)$, constructed in terms of the  inverse 
spectral transforms via (\ref{E:potential-inv-S}), coincides  with the Cauchy data $v_0(x,y)$ for the direct 
spectral transform:
\begin{equation}
v(x,y,0)=v_0(x,y),\label{E:initial-v-S}
\end{equation}

\end{enumerate}
\end{theorem}
\begin{proof} 
\begin{enumerate}
\item
The reality condition (\ref{E:potential-regularity-S}) follows from the fact that the inverse 
scattering equation (\ref{E:Manakov-Santini}) is real for real $\lambda$, and (\ref{E:potential-inv-S}) 
has real coefficients. 

By differentiating (\ref{E:potential-inv-S}) we obtain

\begin{gather}
v_{\alpha}(x,y,t)=-\frac 1{\pi}\int_{\mathbb R}\partial_{\xi}\chi_{-I}(\psi,\lambda) 
\psi_{\alpha} d\lambda,
\label{E:potential-inv-1-S}\\
v_{\alpha\beta}(x,y,t)=-\frac 1{\pi}\int_{\mathbb R}\left[ \partial_{\xi\xi}\chi_{-I}(\psi,\lambda) 
\psi_{\alpha}\psi_{\beta} + \partial_{\xi}\chi_{-I}(\psi,\lambda) \psi_{\alpha\beta}\right] d\lambda, 
\label{E:potential-inv-1-S-bis}
\end{gather}
$\alpha,\beta\in\{x,y,t\}$. Using the properties (\ref{eq:xi-est-6-bis}), it follows that the only integral requiring regularization is the integral for $v_t$. 
This means that, for $t=0$, the function $v_t$ may be discontinuous. 

\item Let $t>0$. We have
\begin{equation}
v_{t}(x,y,t)=-\frac 1{\pi}\int_{\mathbb R}\partial_{\xi}\chi_{-I}(x-\lambda y -\lambda^2 t+ \omega(x,y,t,\lambda),
\lambda) \psi_{t} d\lambda
\end{equation}
and $\omega$ is a bounded function of $\lambda$, 
therefore the convergence of integral immediately follows from Proposition~\ref{P:decay}.

\item To calculate $L\psi$, $M\psi$ we use the following simple formula. 
Let $f(\lambda)$ be a function such that 
$f(\lambda)\in L^p(d\lambda)$, $\lambda f(\lambda)\in 
L^p(d\lambda)$, $1<p<\infty$. Then 
\beq
\lambda H_{\lambda}[f(\lambda)]=  H_{\lambda}[\lambda f(\lambda)]+\frac{1}{\pi}
\int_{\RR}f(\lambda) d\lambda.
\eeq
Applying L to (\ref{E:Manakov-Santini}) we obtain:
\begin{gather}
L\psi= v_x - L(\chi_{-R}(\psi,\lambda)) + 
(\partial_y + \lambda\partial_x +v_x\partial_x) H_{\lambda}[\chi_{-I}(\psi,\lambda))]=
\nonumber\\
=v_x - \partial_{\xi}\chi_{-R}(\psi,\lambda) L\psi  +  
H_{\lambda}[L\chi_{-I}(\psi,\lambda))]+\frac{1}{\pi} \int_{\RR} \partial_x \chi_{-I}(\psi,\lambda)) d\lambda=\\
=v_x - \partial_{\xi}\chi_{-R}(\psi,\lambda) L\psi  +  
H_{\lambda}[\partial_{\xi}\chi_{-I}(\psi,\lambda)) L\psi ]-v_x.\nonumber
\end{gather}
We obtain that $L\Psi\in L^2(d\lambda)$ and solves the homogeneous equation; therefore,
by Lemma~\ref{L:gn-inv-S}, it is a zero element of $L^2(d\lambda)$. 

Analogously,
\begin{gather}
M\psi= \lambda v_x -v_y - \partial_{\xi}\chi_{-R}(\psi,\lambda) M\psi  +  
H_{\lambda}[\partial_{\xi}\chi_{-I}(\psi,\lambda)) M\psi ]+\nonumber\\
+\lambda\frac{1}{\pi}\int_{\RR} \partial_x \chi_{-I}(\psi,\lambda)) d\lambda+
\frac{1}{\pi}\int_{\RR} \lambda\partial_x \chi_{-I}(\psi,\lambda)) d\lambda+
v_x\frac{1}{\pi}\int_{\RR} \partial_x \chi_{-I}(\psi,\lambda)) d\lambda=\nonumber\\
=  \partial_{\xi}\chi_{-R}(\psi,\lambda) M\psi  + 
H_{\lambda}[\partial_{\xi}\chi_{-I}(\psi,\lambda)) M\psi ]+\\
+\lambda v_x-v_y-\lambda v_x- v_x^2+ \frac{1}{\pi}\int_{\RR} \lambda\partial_x \chi_{-I}(\psi,\lambda)) d\lambda\nonumber
\end{gather}
Taking into account that
\beq 
\lambda\partial_x \chi_{-I}(\psi,\lambda)) = - \partial_y \chi_{-I}(\psi,\lambda))
- v_x \partial_x \chi_{-I}(\psi,\lambda)),
\eeq
we obtain that $M\psi\in L^2(d\lambda)$ and 
\beq
M\psi=  \partial_{\xi}\chi_{-R}(\psi,\lambda) M\psi + 
H_{\lambda}[\partial_{\xi}\chi_{-I}(\psi,\lambda)) M\psi ];
\eeq
therefore
$$
M\psi =0.
$$
\item This property is exactly equivalent to the inverse problem equation  
(\ref{E:Manakov-Santini}).
\item
From (\ref{E:analytic-1}) it follows that
\begin{equation}\label{E:jump-inv-S}
\Psi^+(x, y, t,\lambda)-\Psi^-(x, y, t,\lambda)=-2i\chi_{-I}(\psi(x,y,t,\lambda),\lambda),\ \lambda\in\mathbb R
\end{equation}
in $L^2(\RR,d\lambda)$. The standard solution of the Riemann factorization problem in 
terms of the Cauchy integral immediately gives us (\ref{E:eigenfunction-inv-S}).

Combining Theorem~\ref{T:uniqueness-S}  and the H$\ddot{\textrm o}$lder inequality,  we obtain (\ref{E:eigen-condition-x-S}) for fixed  $\lambda\in \CC^\pm\backslash\RR$. 

\item
Finally,  restricting (\ref{E:eigenfunction-inv-S}) to $t=0$, (\ref{E:jump-inv-S}) yields (\ref{E:scattering}) and (\ref{E:Manakov-Santini-0}). So $\psi(x,y,0,\lambda)=\varphi(x,y,\lambda)$, $\Psi(x,y,0,\lambda)=\Phi(x,y,\lambda)$. Comparing 
(\ref{eq:phi-rieamann}), (\ref{eq:v-rieamann}) with   (\ref{E:eigenfunction-inv-S}),
(\ref{E:potential-inv-S}) we obtain (\ref{E:initial-v-S}).
\end{enumerate}
\end{proof}

\begin{theorem}\label{T:inverse-problem-P-1-S} 
Suppose $v_0(x,y)\in\mathfrak S_{x,y}$ with compact support and satisfies the sufficiently small condition from
Definition~\ref{def:small-data}. Then the Cauchy problem of the Pavlov equation
\begin{equation}\label{E:cauchy-pavlov-S}
\begin{split}
&v_{xt}+v_{yy}=v_yv_{xx}-v_xv_{xy},\quad \forall x,\,y\in\RR,\,t\in\mathbb R^+,\\
&v(x,y,0)=v_0(x,y) 
\end{split}
\end{equation}
admits a real solution $v=v(x,y,t)$ such that $v,\,v_x,\,v_y,\,v_{xx},\,v_{xy},v_{xt},\,v_{yy} 
\in C(\RR\times\RR\times \RR^+)\cap L^\infty(\RR\times\RR\times \RR^+)$.
\end{theorem}
\begin{proof} Applying Proposition \ref{P:decay} and computing the compatibility of the Lax pair (\ref{E:pavlov-inv-x-S}) and (\ref{E:pavlov-inv-y-S}), we obtain
\[
\left(v_{xt}+v_{yy}-v_yv_{xx}+v_xv_{xy}\right)\partial_x\Psi\equiv 0.
\]
Hence we obtain  (\ref{E:cauchy-pavlov-S}) by (\ref{E:eigen-condition-x-S}).
\end{proof}

\section{Summary of the results and concluding remarks}\label{S:summary}
We have shown that the direct problem consists of the following steps:
\begin{itemize}
\item From the potential $v(x,y)$ we construct the scattering data $\sigma(\xi,\lambda)$, 
solving the ODE (\ref{ODE}).
\item From the scattering data  $\sigma(\xi,\lambda)$ we construct the spectral data 
$\chi(\xi,\lambda)$ solving the shifted Riemann problem (\ref{E:shift-intro}).
\end{itemize}
These two steps do not require small norm assumptions.

The inverse problem consists of the following two steps:
\begin{itemize}
\item From the spectral data $\chi(\xi,\lambda)$ we construct the real Jost eigenfunctions
solving the nonlinear integral equation (\ref{inversion1}), under the small norm assumption. 
\item From the real eigenfunctions we construct the potential $v(x,y,t)$ using formula
(\ref{E:potential-inv-S}). 
\end{itemize}

The following remark is important.

\begin{remark}
A careful reader may notice that the above basic steps do not involve explicitly the analytic eigenfunctions; therefore, strictly speaking, the Cauchy problem for the Pavlov equation can be solved without introducing them. However, their existence pervades the whole IST. Indeed, not only it is crucial in motivating the shifted Riemann problem  (\ref{E:shift-intro}) of the direct problem, but it is also equivalent to the nonlinear integral equation (\ref{inversion1}) of the inverse problem. 
\end{remark}

\section{The analytic estimates} \label{S:estimates}
In this section we present the proofs of some of the analytical estimates we use in our paper.  

\textbf{Proof of Proposition~\ref{L:direct-sigma-asy-00}.}

The main tool for proving these estimates in the Gronwall's inequality.  By definition,
$$
\sigma(\tau,\lambda)=\lim\limits_{y\rightarrow+\infty} h(y,\tau,\lambda)-\tau,
$$
where $h=h(y,\tau,\lambda)$ denotes the solution of the vector field ODE:
\begin{equation}
\label{eq:xi_tau_0}
\frac {d h}{dy}=v_x(h+\lambda y,y,\lambda), 
\end{equation}
with the boundary condition:
$$
\lim\limits_{y\rightarrow-\infty} h(y,\tau,\lambda)=\tau.
$$
Therefore:
$$
\sigma(\tau,\lambda)=\int\limits_{-\infty}^{+\infty} v_x(h( y,\tau,\lambda)+\lambda y, y) d y,
$$
and
$$
|\sigma(\tau,\lambda)|\le \int\limits_{-\infty}^{+\infty} \left[\max\limits_{x\in\RR} |v_x(x, y)|  \right] d y=B_0,
$$
The function $h_{\tau}$ satisfies the linearized equation:
\begin{equation}
\label{eq:xi_tau_1}
\frac{dh_{\tau}}{dy} = v_{xx}\ h_{\tau}
\end{equation}
with the boundary value
$$
\lim\limits_{y\rightarrow-\infty} h_{\tau}(y,\tau,\lambda)=1.
$$
Equation~(\ref{eq:xi_tau_1}) can be written as:
$$
\frac{d}{d y}\log(h_{\tau}) = v_{xx}(h(y,\tau,\lambda)+\lambda y,y),
$$
therefore 
$$
|\log(h_{\tau}(y,\tau,\lambda ))|\le \int\limits_{-\infty}^{+\infty}  \left[\max\limits_{x\in\RR} |v_{xx}(x, y)|\right] d y,
$$
$$
\exp\left(-\left[\max\limits_{x\in\RR} |v_{xx}(x, y)|\right] d y\right)-1\le h_{\tau}(y,\tau,\lambda )-1\le  
\exp\left(\left[\max\limits_{x\in\RR} |v_{xx}(x, y)|\right] d y\right)-1=B_1,
$$
and
$$
|h_{\tau}(y,\tau,\lambda )-1|\le B_1, \ \ |h_{\tau}(y,\tau,\lambda )|\le B_1+1,
$$
which automatically implies the necessary estimate on $ |\sigma_{\tau}(\tau,\lambda)|$.

The next step is to estimate the solutions of the equation for $h_{\tau\tau}$
\begin{equation}
\label{eq:xi_tau_2}
\frac{dh_{\tau\tau}}{dy} = v_{xxx}\ h_{\tau}^2 + v_{xx}\ h_{\tau\tau} 
\end{equation}
with the boundary condition:
$$
\lim\limits_{y\rightarrow-\infty} h_{\tau\tau}(y,\tau,\lambda)=0.
$$
We have an inhomogeneous linear equation; therefore we can use the standard estimate:
$$
|h_{\tau\tau}(y,\tau,\lambda )|\le \left[ \int\limits_{-\infty}^{+\infty} \left[\max\limits_{x\in\RR} |v_{xxx}(x, y)  |\right] d y \right]\cdot
[\max\limits_{y,\tau\in\RR} |h_{\tau}(y,\tau,\lambda)|^2]\cdot  
$$
$$
\cdot\exp\left(\int\limits_{-\infty}^{+\infty}  \left[\max\limits_{x\in\RR} |v_{xx}(x, y)|\right] d y\right),
$$
which implies the estimate on $ |\sigma_{\tau\tau}(\tau,\lambda)|$.
Equation for $h_{\tau\tau\tau}$ has the form 
\begin{equation}
\label{eq:xi_tau_3}
\frac{dh_{\tau\tau\tau}}{dy}= v_{xxxx}\ h_{\tau}^3+
3 v_ {xxx}\ h_{\tau}\ h_{\tau\tau } + v_ {xx}\ h_{\tau\tau\tau}
\end{equation}
with the boundary condition:
$$
\lim\limits_{y\rightarrow-\infty} h_{\tau\tau\tau}(y,\tau,\lambda)=0.
$$
Again we can estimate the function  $ |\sigma_{\tau\tau\tau}(\tau,\lambda)|$ as product of the integral of the modulus of the inhomogeneous term times 
the exponent of the modulus of the homogeneous coefficient:
$$
|h_{\tau\tau\tau}(y,\tau,\lambda )|\le \left( \left[ \int\limits_{-\infty}^{+\infty} \left[\max\limits_{x\in\RR} |v_{xxxx}(x, y)  |\right] d y 
\right]\cdot [\max\limits_{y,\tau\in\RR} |h_{\tau}(y,\tau,\lambda)|^3]+ \right.
$$
$$
\left. +{\color{Green}3}\left[ \int\limits_{-\infty}^{+\infty} \left[\max\limits_{x\in\RR} |v_{xxx}(x, y)  |\right] d y 
\right]\cdot [\max\limits_{y,\tau\in\RR} |h_{\tau}(y,\tau,\lambda)|] \cdot [\max\limits_{y,\tau\in\RR} |h_{\tau\tau}(y,\tau,\lambda)|] \right) \cdot
$$
$$
\cdot\exp\left(\int\limits_{-\infty}^{+\infty}  \left[\max\limits_{x\in\RR} |v_{xx}(x, y)|\right] d y\right),
$$
which implies the estimate on $ |\sigma_{\tau\tau\tau}(\tau,\lambda)|$.

Let us denote $h(y,\tau,\lambda)=\tau+\tilde h(y,\tau,\lambda)$. Equations (\ref{eq:xi_tau_0}),  (\ref{eq:xi_tau_1}) can be interpreted as ODEs for the functions $\tilde h(y,\tau,\lambda)$, $\tilde h_{\tau}(y,\tau,\lambda)$ in the Hilbert space $L^2(d\tau)$. We obtain:
\begin{equation}
\label{eq:xi_tau_0_bis}
\frac {d \tilde h}{dy}=v_x(\tilde h+\tau+\lambda y,y), 
\end{equation}
\begin{equation}
\label{eq:xi_tau_1_bis}
\frac{d\tilde h_{\tau}}{dy} = v_{xx}(\tilde h+\tau+\lambda y,y)+  v_{xx}(\tilde h+\tau+\lambda y,y) \tilde h_{\tau}
\end{equation}
We see that
$$
\| \tilde h(y,\tau,\lambda)\|_{L^2(d\tau)} \le \int\limits_{-\infty}^{+\infty} \| v_x(\tilde h+\tau+\lambda y,y) \|_{L^2(d\tau)} dy,
$$
$$
\| \tilde h _{\tau}(y,\tau,\lambda)\|_{L^2(d\tau)} \le \left(\int\limits_{-\infty}^{+\infty} \| v_{xx}(\tilde h +\tau+\lambda y,y) \|_{L^2(d\tau)} dy\right)
\cdot\exp\left(\int\limits_{-\infty}^{+\infty}  \left[\max\limits_{x\in\RR} |v_{xx}(x,\tilde y)|\right] d\tilde y\right),
$$
We assume now, that $B_1<1$. We know, that 
$$
\| \ldots \|_{L^2(d\tau)} \le  \|\ldots \|_{L^2(dx)} \cdot \sqrt{\max \frac{d\tau}{dx}} = \frac{\|\ldots \|_{L^2(dx)}}{\sqrt{\min \frac{dx}{d\tau}}},  
$$ 
for a fixed $y$, $\lambda$, but $\frac{dx}{d\tau}=h_{\tau}(\tau,y,\lambda)$, therefore 
$$
\| v_x(\tilde h +\tau+\lambda y,y) \|_{L^2(d\tau)}\le \frac{1}{\sqrt{1-B_1}}\cdot   \| v_x(x,y)\|_{L^2(d x)},
$$
$$
\| v_{xx}(\tilde h +\tau+\lambda y,y) \|_{L^2(d\tau)}\le \frac{1}{\sqrt{1-B_1}}\cdot   \| v_{xx}(x,y)\|_{L^2(d x)},  
$$
which completes the proof.

\textbf{Proof of Proposition~\ref{L:direct-sigma-asy-0}.}

Due to Definition~\ref{def:sigma}, it is sufficient to prove the Lemma for $|\lambda|\gg 1$. Thus we always assume $|\lambda|\gg 1$ in the following proof.
{\color{Green} The cases $\lambda\rightarrow+\infty$ and $\lambda\rightarrow-\infty$ are completely analogous, therefore we assume now that $\lambda\rightarrow+\infty$.}

Let us rewrite the definition of the scattering data using $\xi$ and $x$ as new coordinates 
on the $(x,y)$-plane. 
The $x$-coordinate is expressed through $\xi$, $y$ using the following formulas:
\beq\label{E:change}
x=\tilde h(y;\xi,-\infty,\lambda) +\lambda y=
\xi+\lambda y+\int_{-\infty}^{y} v_x( \xi+\lambda y'+h(y';\xi,-\infty,\lambda),y') dy'.
\eeq
From the implicit function theorem, this map can be inverted with respect to $y$: 
$$
y=H(\xi,x,\lambda)=\frac{x-\xi}{\lambda}+\frac{H_1(\xi,x,\lambda)}{\lambda^2} , \ \ 
H_1=\mathcal O(1)
$$
where
\begin{equation}
\label{ODE-bis}
\frac{\partial y}{\partial x}=\frac{\partial H(\xi,x,\lambda)}{\partial x}=
\frac{1}{\lambda+v_x(x,H(\xi,x,\lambda))},
\end{equation}
or, equivalently
\begin{equation}
\label{ODE-ter}
\frac{\partial H_1(\xi,x,\lambda)}{\partial x}=-
\frac{v_x\left(x,\frac{x-\xi}{\lambda}+\frac{H_1(\xi,x,\lambda)}{\lambda^2}\right) }
{1+v_x\left(x,\frac{x-\xi}{\lambda}+\frac{H_1(\xi,x,\lambda)}{\lambda^2}\right)/\lambda},
\end{equation}
$$
H_1(\xi,-\infty,\lambda)=0.
$$
We see, that
$$
\sigma(\xi,\lambda)=-\frac{H_1(\xi,+\infty,\lambda)}{\lambda}.
$$
Let us denote
$$
\mathring{H_1}(\mathring{\xi},x,{\mathring\lambda})=H_1(\mathring{\xi}/{\mathring\lambda},x,1/{\mathring\lambda}).
$$
Taking into account that
$$
{\mathring\lambda}=\frac{1}{\lambda}, \ \ \mathring{\xi} = \frac{\xi}{\lambda},
$$
we obtain
\begin{equation}
\label{ODE-quater}
\frac{\partial \mathring{H_1}(\mathring{\xi},x,{\mathring\lambda})}{\partial x}=-
\frac{v_x\left(x,-\mathring{\xi}+{\mathring\lambda} x+ {\mathring\lambda}^2 \mathring{H_1}(\mathring{\xi},x,{\mathring\lambda}) \right) }
{1+ {\mathring\lambda} v_x\left(x,-\mathring{\xi}+{\mathring\lambda} x+ {\mathring\lambda}^2 
\mathring{H_1}(\mathring{\xi},x,{\mathring\lambda}) \right)},
\end{equation}
For $|{\mathring\lambda}|<\frac{1}{2\max |v_x(x,y)|}$ the right-hand side of (\ref {ODE-quater}) 
is smooth in $\mathring{\xi}$, ${\mathring\lambda}$. We solve this equation in the finite interval 
$-D_x\le x \le D_x$; therefore $\mathring{H_1}(\mathring{\xi},+\infty,{\mathring\lambda})= 
\mathring{H_1}(\mathring{\xi},D_x,{\mathring\lambda})$ smoothly depends on the parameters. It is easy to check that, for 
$|-\mathring{\xi}|>D_y+|{\mathring\lambda}|D_x$, the right-hand side of (\ref {ODE-quater}) is identical to 0,
therefore $\mathring\sigma(\mathring{\xi},{\mathring\lambda})\equiv 0$ in the region 
$|\mathring{\xi}|>D_y+|{\mathring\lambda}|D_x$. 

Expanding  (\ref {ODE-quater}) at ${\mathring\lambda}=0$ we obtain:
\begin{equation}
\label{ODE-5}
\frac{\partial \mathring{H_1}(\mathring{\xi},x,{\mathring\lambda})}{\partial x}=-
v_x\left(x,-\mathring{\xi}\right)+ O({\mathring\lambda});
\end{equation}
therefore 
$$
\mathring{H_1}(\mathring{\xi},+\infty,{\mathring\lambda}) = \int\limits_{-D_x}^{D_x} v_x\left(x,-\mathring{\xi}\right) dx + O({\mathring\lambda})=v(D_x,-\mathring{\xi})- v(-D_x,-\mathring{\xi})+O({\mathring\lambda})= O({\mathring\lambda}).
$$ 
From the Hadamard's lemma it follows, that 
$$
\frac{\sigma(\mathring{\xi}/{\mathring\lambda},1/{\mathring\lambda})}{{\mathring\lambda}^2}=  
-\frac{\mathring{H_1}(\mathring{\xi},+\infty,{\mathring\lambda})}{{\mathring\lambda}}
$$
is a regular function of $\mathring{\xi}$, ${\mathring\lambda}$ for sufficiently small ${\mathring\lambda}$. We proved the first part.

To prove the corollary, let us point out that, in the new variables,
$$
\partial_{\lambda}=-{\mathring\lambda}^2\partial_{{\mathring\lambda}}-
{\color{Green} \mathring\lambda\mathring\xi }\partial_{\mathring{\xi}}, \ \ 
\partial_{\xi}={\mathring\lambda} \partial_{\mathring{\xi}}.
$$
Therefore any differentiation of the scattering data with respect to $\lambda$, $\xi$ increases the order of zero with respect to ${\mathring\lambda}$ at the point ${\mathring\lambda}=0$ by one. Taking into account that 
$$
\| \|_{L^2(d\xi)}= \frac{\| \|_{L^2(d\mathring{\xi})}}{\sqrt{|{\mathring\lambda}|}},
$$
we finish the proof.

\textbf{Proof of Theorem~\ref{T:complex-bdry}.}

To prove the Theorem, let us make an appropriate change of variables. It will be 
done in 5 steps. 

\noindent{\underline{\it Step 1 }}: Consider a point $(x,y)\in\RR^2$. Denote by $\hat h(y';x,y,\lambda_R)$ 
the solution of the ordinary differential equation
\beq
\label{eq:L1_22}
\frac{d\hat h}{d y'} = \lambda_R+v_x(\hat h, y') 
\eeq
with the boundary condition
\beq
\label{eq:L1_23}
\hat h(y;x,y,\lambda_R)=x.
\eeq
The first change of variables $\mathcal F_1:(x,y)\rightarrow(x_1,y_1)$   is defined by:
\beq
\label{eq:L1_24}
\left\{
\begin{array}{ll}
x_1=\lim\limits_{y'\rightarrow-\infty} \hat h(y';x,y,\lambda_R)-\lambda_R y'
=\varphi_-(x,y,\lambda_R)=\varphi(x,y,\lambda_R),&  \ \ y <0\\
x_1=\lim\limits_{y'\rightarrow+\infty} \hat h(y';x,y,\lambda_R)-\lambda_R y'
=\varphi_+(x,y,\lambda_R),& 
\ \  y >0\\
y_1 = y & 
\end{array}
\right.
\eeq
Of course the map is discontinuous on the line $y=0$, and
\beq\label{E:ode-tra}
\begin{split}
&\left(\varphi_\pm\right)_y+(\lambda_R+v_x)\left(\varphi_\pm\right)_x=0.
\end{split}
\eeq In the new variables 
we have
\beq
\label{eq:L1_25}
L = \partial_y+(\lambda+v_x)\partial_x=\partial_{y_1}+ i\lambda_I\kappa(x_1,y_1) \partial_{x_1} ,
\eeq
where
\beq\label{E:kappa}
\kappa(x_1,y_1)=\frac{\partial\varphi_\pm}{\partial x}(x,y,\lambda_R)|_{(x,y)=\mathcal F_1^{-1}(x_1,y_1)},\ \ y\ne 0.
\eeq Moreover, there exists a pair of positive constants $ C_1$, $C_2$ such that:
\beq
\label{eq:L1_26}
0< C_1 \le\kappa(x_1,y_1) \le  C_2. 
\eeq
\noindent{\underline{\it Step 2 }}: To investigate the boundary behaviors of the complex eigenfunction, we observe that, for $\lambda=\lambda_R+i\lambda_I$,  
$|\lambda_I|\ll1$, it is natural to conjecture that 
$\Phi(x,y,\lambda)$ is almost constant on the trajectories of the vector 
field
\beq
\label{eq:L1_13}
\hat L\equiv \partial_y+\lambda_R\partial_x+v_x\partial_x.
\eeq
These trajectories are defined by (\ref{eq:L1_22}) and (\ref{eq:L1_23}). 
Hence, if 
\beq
\label{eq:L1_15}
\hat h(x,y,y',\lambda_R)=\xi + \lambda_R\ y' \ \ \mbox{as} \ \ y'\rightarrow-\infty,
\eeq
then 
\beq
\label{eq:L1_16}
\hat h(x,y,y',\lambda_R)=\xi+\sigma(\xi,\lambda_R) + \lambda_R\ y' \ \ \mbox{as} \ \ y'\rightarrow+\infty, 
\eeq
where $\sigma(\xi,\lambda_R)$ is defined by Definition \ref{def:sigma} (see the proof of Lemma \ref{L:pde-system}).

Recall that $z=x-\lambda y$. Assume that the support of $v_x(z,\bar z)$ is located inside the strip
$|z_I|<\varepsilon$, $\varepsilon\ll1$. Then $\Phi(z,\bar z,\lambda)$ 
is holomorphic in $z$ outside a small neighbourhood of the real line and we have
\beq
\label{eq:L1_17}
\Phi(\xi+\sigma(\xi,\lambda)+i\epsilon,\lambda)\sim  \Phi(\xi-i\epsilon,\lambda)\ \ 
\mbox{for} \ \ \lambda_I<0, 
\eeq
\beq
\label{eq:L1_18}
\Phi(\xi+\sigma(\xi,\lambda)-i\epsilon,\lambda)\sim  \Phi(\xi+i\epsilon,\lambda)\ \ 
\mbox{for} \ \ \lambda_I>0. 
\eeq

Consider the Riemann-Hilbert problem with shift (\ref{E:shifted-RH}), or, via function
\beq\label{E:w-chi}
w(\xi,\lambda_R)=\xi+\chi(\xi,\lambda_R),
\eeq 
\begin{gather}
\label{eq:L1_19}
w(\xi+\sigma(\xi)+i0,\lambda_R) =  w(\xi-i0,\lambda_R),\ \ \xi\in\RR,\\
w(z) = z + o(1) \ \ \mbox{as} \ \ z\rightarrow\infty.\nonumber
\end{gather}
Then the hypothetical formulas for $\Phi^-(x,y,\lambda_R)$,  
$\Phi^+(x,y,\lambda_R)$ read: 
\begin{gather}
\label{eq:L1_20}
\Phi^-(x,y,\lambda_R)=w(\varphi_-(x,y,\lambda_R)-i0,\lambda_R)= 
w(\varphi_+(x,y,\lambda_R)+i0,\lambda_R) \\ 
\Phi^+(x,y,\lambda_R)=\overline{\Phi^-(x,y,\lambda_R)}.\nonumber
\end{gather}

\noindent{\underline{\it Step 3 }}: Assume $\lambda_I<0$ from now on. 
Let us use the following rescaling:  $\mathcal F_2:(x_1,y_1)\rightarrow(x_2,y_2)$
\beq
\label{eq:L1_27}
\left\{
\begin{array}{l}
x_2=x_1\\
y_2 = \lambda_I y_1 \\
z_2=x_2-i y_2.
\end{array}
\right.
\eeq
In the new variables 
\beq
\label{eq:L1_28}
L =\lambda_I(\partial_{y_2}+i  \kappa\left(x_2,\frac{y_2}{\lambda_I}\right) \partial_{x_2}).
\eeq

\noindent{\underline{\it Step 4 }}: Let us define a new complex variable $z_3$, 
$\mathcal F_3:(x_2,y_2)\rightarrow z_3$  by
\beq
\label{eq:L1_29}
z_3=x_2-iy_2+\chi(x_2-iy_2,\lambda_R),
\eeq
where $\chi(\xi,\lambda)$ is the solution of the shifted Riemann-Hilbert 
problem (\ref{E:shifted-RH}) (existence of the solution  is proved in 
\cite{Ga66}). Note that the composition $\mathcal F_3\circ\mathcal F_2\circ\mathcal F_1$ is continuous by the property:
\beq
\label{eq:L1_30}
\mbox{If} \ \ \mathcal F_1(x,-0)=(\xi,0) \ \ \mbox{then} \ \  \mathcal F_1(x,+0)=(\xi+\sigma(\xi,\lambda_R),0).
\eeq
Consequently, (\ref{E:lax-complex}) takes the form 
\beq
\label{eq:L1_32}
[\partial_{\bar z_3} + q(z_3,\bar z_3, \lambda) \partial_{z_3}] \Phi = 0
\eeq
where
\beq
\label{eq:L1_33}
|q(z_3,\bar z_3,\lambda)|<\mathcal C_3[v](\lambda_R)<1,
\eeq
the support of $q(z,\bar z,\lambda)$ has area of order $O(\lambda_I)$.

It is natural to consider Beltrami equation (\ref{eq:L1_32}) in the space 
$L^{2+\epsilon}(dz_3d\bar z_3)\cap L^{2-\epsilon}(dz_3d\bar z_3)$ where $\epsilon$ is 
sufficiently small.
Again we can write 
\beq
\label{eq:L1_34}
\Phi(z_3,\bar z_3,\lambda) = z_3 + 
\partial^{-1}_{\bar z_3} \alpha(z_3,\bar z_3,\lambda)
\eeq
where
\beq
\label{eq:L1_35}
[1+q(z_3,\bar z_3,\lambda)\partial_{z_3}\partial^{-1}_{\bar z_3}]
\alpha(z_3,\bar z_3,\lambda)+q(z_3,\bar z_3,\lambda)=0.
\eeq

Taking into account (\ref{eq:L1_33}) we see that 
\beq
\label{eq:L1_36}
|\alpha(z_3,\bar z_3,\lambda) |_{L^{p}} = O(\lambda_I),\ \ 
|p-2|<\epsilon.
\eeq
Using the estimates from \cite{V62} we see, that 
$$
\|\Phi(z_3,\bar z_3,\lambda) - z_3 \|_{L^{\infty}(dz_3d\bar z_3)}=O(\lambda_I),
$$
and $\Phi(z_3,\bar z_3,\lambda)$ uniformly converges to $z_3$. 

\noindent{\underline{\it Step 5 }}: Consider the function $\Phi(x,y,\lambda)$ on
the line $y=y_0<-D_y$. We see, that 
$$
z_2 = \xi + i |\lambda_I| y_0, \ \ \mbox{where} \ \ \xi=x-\lambda_R y,
$$
therefore
$$
\Phi(x,y_0,\lambda-i 0)= \left. z_3(z_2)\vphantom{\int} \right|_{z_2=\xi-i0}=
\xi+\chi_{-}(\xi,\lambda). 
$$
On this line 
$$
\phi_{-}(x,y,\lambda)=\xi,
$$
therefore
$$
\Phi(x,y,\lambda-i 0)= \phi_{-}(x,y,\lambda)+
\chi_{-}(\phi_{-}(x,y,\lambda),\lambda). 
$$
The proof is completed.

\textbf{Proof of Lemma~\ref{lem:xi-est-1}.}

To start with, let us point out that
$$
f(t,\tau)=\partial_{\tau} \log{\hat s(t,\tau)}, 
$$
where
$$
\hat s(t,\tau)=\frac{s(t)-s(\tau)}{t-\tau}=\int\limits_{0}^{1} 
s'(\alpha t+[1-\alpha]\tau)d\alpha,
$$
$$
\partial^k_t \partial^l_{\tau} \hat s(t,\tau) = \int\limits_{0}^{1} 
\alpha^k [1-\alpha]^l s^{(k+l+1)} (\alpha t+[1-\alpha]\tau)d\alpha.
$$
Therefore
$$
|\partial^k_t \partial^l_{\tau} \hat s(t,\tau)|\le\max\limits_{\xi}
|s^{(k+l+1)}(\xi) |.  
$$
We see that, if the corresponding derivatives exist,
$$
f=\frac{\hat s_{\tau}}{\hat s}, \ \ 
f_t=\frac{\hat s_{t \tau}}{\hat s}- \frac{\hat s_{t} \hat s_{\tau}}{\hat s^2}, 
\ \ 
f_{tt}=\frac{\hat s_{tt\tau}}{\hat s}- 2\frac{\hat s_{t} \hat s_{t\tau}}{\hat s^2} - \frac{\hat s_{tt} \hat s_{\tau}}{\hat s^2}+ 2\frac{\hat s^2_{t} \hat s_{\tau}}{\hat s^3},
$$
and
\begin{equation}
\label{eq:hadamar1}
\begin{split}
|f(t,\tau)|\le \frac{\max |s''|}{\min|s'|},  \ \ 
|f_t(t,\tau)|\le \frac{\max |s'''|}{\min|s'|}+ \frac{\max |s''|^2}{\min|s'|^2},\\
|f_{tt}(t,\tau)|\le \frac{\max |s''''|}{\min|s'|}+ \frac{3\max |s'''| \max |s''|}{\min|s'|^2}+  \frac{2 \max |s''|^3}{\min|s'|^3}.
\end{split}
\end{equation}

We know that 
$$
\|K\|_{L^{\infty}} =\frac{1}{2\pi} \max\limits_{t\in\RR} \int_{\RR} |f(t,\tau)| d\tau=\frac{1}{2\pi}
[I_1(t)+I_2(t)],
$$
$$
I_1(t)=\int\limits_{|\tau-t|\le 1} \left|\frac {s'(\tau)}{s(\tau)-s(t)}-\frac 1{\tau-t}\right| d\tau,
$$
$$
I_2(t)=\int\limits_{|\tau-t|\ge 1} \left|\frac {s'(\tau)}{s(\tau)-s(t)}-\frac 1{\tau-t}\right| d\tau\le I_{21}+I_{22},
$$
$$
I_{21}(t)=\int\limits_{|\tau-t|\ge 1} \left|\frac {1}{s(\tau)-s(t)}-\frac 1{\tau-t}\right| d\tau ,
$$
$$
I_{22}(t)=\int\limits_{|\tau-t|\ge 1} \left|\frac {\sigma'(\tau)}{s(\tau)-s(t)}\right| d\tau,
$$
From (\ref{eq:hadamar1}) we see, that
$$
|f(t,\tau)|\le \frac{\color{Green}C_2}{1-C_1}\le {\color{Green}2} C_2,
$$
and 
$$
I_1 \le\int\limits_{|\tau-t|\le 1} {\color{Green}2} C_2 d\tau = {\color{Green}4}C_2. 
$$
Let us estimate now $I_{21}$. We have $s(\tau)-s(t)=\tau-t+\sigma(\tau)-\sigma(t)$. We assumed that $|\sigma(\tau)-\sigma(t)|\le2C_0\le1/2$; therefore
$$
\left|\frac {1}{s(\tau)-s(t)}-\frac 1{\tau-t}\right| = 
\left|\frac 1{\tau-t}\right| \  \left|\frac {1}
{1+\frac{\sigma(\tau)- \sigma(t)}{\tau-t} }-1 \right| \le
\frac{4C_0}{(\tau-t)^2},
$$
and
$$
I_{21}\le 8 C_0.
$$
To estimate $I_{22}$, we use the H\"older inequality
$$
I_{22}(t)\le\sqrt{ \int\limits_{|\tau-t|\ge 1} |\sigma'(\tau)|^2 d\tau}\ \cdot \ \sqrt{ \int\limits_{|\tau-t|\ge 1} \frac {1}{(s(\tau)-s(t))^2} d\tau}\le 
$$
$$
\le \|\sigma'(\tau)\|_{L^2(d\tau)} \cdot \ \sqrt{ \int\limits_{|\tau-t|\ge 1} \frac {4}{(\tau-t)^2} d\tau}=\sqrt{8} \hat C_1.
$$
Combining estimates for $I_1$, $I_{21}$, $I_{22}$ we complete the proof of the first part.

To prove the second part, we use the standard estimate:
$$
|h'_2(t)|\le \left[\max\limits_{t\in\RR} \int\limits_{\tau\in\RR}
\left|\partial_t f(t,\tau) \right| d\tau \right] \cdot  \| h_1(t)\|_{L^{\infty}(dt)}.
$$
We have:
$$
\int\limits_{\tau\in\RR} \left|\partial_t f(t,\tau) \right| d\tau \le I_1+I_2,
$$
where
$$
I_1=\int\limits_{|\tau-t|\le1} \left|\partial_t f(t,\tau) \right| d\tau,
$$
$$
I_2=\int\limits_{|\tau-t|\ge1} \left|\partial_t f(t,\tau) \right| d\tau.
$$
From (\ref{eq:hadamar1}), we see that
$$
I_1\le 2 \left[\frac{C_3}{1-C_1}+\frac{C_2^2}{(1-C_1)^2}\right]\le 4 C_3 +8 C_2^2.
$$
Let us introduce the following notation:
$$
y=\hat O(x), \ \mbox{if} \ \ |y|\le|x|.
$$
Let us estimate $I_2$. We have:
$$
\partial_t f(t,\tau)=\frac{s'(\tau)s'(t)}{(s(\tau)-s(t))^2}-
\frac{1}{(\tau-t)^2}=
$$
$$
=\frac{1}{(s(\tau)-s(t))^2}+
\frac{\sigma'(\tau)+ \sigma'(t)}{(s(\tau)-s(t))^2}
+\frac{\sigma'(\tau) \sigma'(t)}{(s(\tau)-s(t))^2}
-\frac{1}{(\tau-t)^2}
$$
By definition, 
$$
s(\tau)-s(t)= \tau - t + \sigma(\tau)-\sigma(t)=
(\tau-t)\left[1+\hat O\left(C_1 \right)  \right]
$$
$$
\frac{1}{(s(\tau)-s(t))^2}=\frac{1}{(\tau-t)^2}
\left[1+ \hat O\left(6 C_1) \right) \right] =
\hat O \left(  \frac{4}{(\tau-t)^2} \right).
$$
Therefore
$$
|I_2|\le \int\limits_{|\tau-t|\ge1} 
\left[  \frac{6C_1}{(\tau-t)^2} + \frac{8C_1}{(\tau-t)^2} +
\frac{4C_1^2}{(\tau-t)^2} \right] d\tau =28 C_1 +8 C_1^2.
$$
Finally we obtain 
$$
\int\limits_{\tau\in\RR} \left|\partial_t f(t,\tau) \right| d\tau \le   
4 C_3+ 8 C_2^2 + 28 C_1 +8 C_1^2.
$$

\textbf{Proof of Lemma~\ref{lem:xi-est-3}.}

We have
$$
g(\xi)=g_1(\xi)+g_2(\xi),
$$
where
$$
g_1(\xi)= -\frac 12\sigma(\xi),
$$
$$
g_2(\xi)=\frac 1{2\pi i}\int_\RR \frac {\sigma(s^{-1}(\xi''))}{\xi'' -s(\xi)} d\xi'',
$$
or, equivalently,
$$
g_2(s^{-1}(\eta))=\frac 1{2\pi i}\int_\RR \frac {\sigma(s^{-1}(\eta''))}{\eta'' -\eta} d\eta'',
$$
where $s^{-1}(\eta)=s^{-1}(\eta,\lambda)$ denotes the inversion of the function $\eta=s(\xi,\lambda)$ with respect to $\xi$:
$$
s(s^{-1}(\eta,\lambda),\lambda)\equiv\eta.
$$
Let us denote:
$$
\hat g_2(\eta)= g_2(s^{-1}(\eta)), \ \ \hat\sigma(\eta'') = \sigma(s^{-1}(\eta'')).
$$
We have:
$$
\hat g_2(\eta)=\frac 1{2\pi i}\int_\RR \frac {\hat\sigma(\eta'')}{\eta'' -\eta} d\eta'', \
\hat g_{2,\eta}(\eta)=\frac 1{2\pi i}\int_\RR \frac {\hat\sigma_{\eta''}(\eta'')}{\eta'' -\eta} d\eta'',
$$
$$
2\pi|\hat g_2(\eta)|=\left| \int\limits_{|\eta''-\eta|\le1} \frac {1}{\eta'' -\eta}\hat\sigma(\eta'')d\eta''+
\int\limits_{|\eta''-\eta|\ge1} \frac {1}{\eta'' -\eta}\hat\sigma(\eta'')d\eta'' \right|\le I_1 + I_2,
$$
where
$$
I_1 = \left|  \int\limits_{|\eta''-\eta|\le1} \frac {\hat\sigma(\eta'')- \hat\sigma(\eta) }{\eta'' -\eta} d\eta''  \right|\le 2 
\max |  \hat\sigma_{\eta}(\eta) |,
$$
and $I_2$ can be estimated using the  H\"older inequality
$$
I_2 = \left|  \int\limits_{|\eta''-\eta|\ge1} \frac {\hat\sigma(\eta'')- \hat\sigma(\eta) }{\eta'' -\eta} d\eta''  \right| \le 
\| \hat\sigma_{\eta} \|_{L^2(d\eta)}\cdot \sqrt{\int\limits_{|\eta''-\eta|\ge1} \frac {1}{(\eta'' -\eta)^2} d\eta' } = 
\sqrt{2}\| \hat\sigma\|_{L^2(d\eta)}.
$$
We obtained:
$$
|\hat g_2(\eta)| \le\frac{1}{2\pi}\cdot[2 \max | \hat\sigma_{\eta}(\eta)| + \sqrt{2}\| \hat\sigma\|_{L^2(d\eta)} ].
$$
Similarly:
$$
|\hat g_{2,\eta}(\eta)| \le\frac{1}{2\pi}\cdot[2 \max |  \hat\sigma_{\eta\eta}(\eta)| + \sqrt{2}\| \hat\sigma_{\eta}\|_{L^2(d\eta)} ].
$$
We have
$$
\hat\sigma_{\eta}(\eta)=\frac{d\xi}{d\eta}\cdot\sigma_{\xi}(s^{-1}(\eta)), \ \ 
\hat\sigma_{\eta\eta}(\eta)=\left(\frac{d\xi}{d\eta}\right)^2\cdot\sigma_{\xi\xi}(s^{-1}(\eta))+
\frac{d^2\xi}{d\eta^2}\cdot\sigma_{\xi}(s^{-1}(\eta)),
$$
$$
\frac{d\xi}{d\eta}=\left(\frac{d\eta}{d\xi}\right)^{-1}, \ \
\frac{d^2 \xi}{d\eta^2}=-\frac{d^2 \eta}{d\xi^2}\cdot \left(\frac{d\eta}{d\xi}\right)^{-3}, \ \ 
d\eta = \frac{d\eta}{d\xi}\cdot d\xi
$$
We assumed that $B_1<\frac{1}{2}$; therefore 
\begin{equation}
\label{eq:der-est-1}
\frac{1}{2}\le \frac{d\xi}{d\eta}\le 2, \ \ \frac{1}{2}\le \frac{d\eta}{d\xi}\le 2, \ \ 
\left|\frac{d^2 \xi}{d\eta^2}\right|\le 8 |\sigma_{\xi\xi}|,
\end{equation}
$$
|\hat\sigma_{\eta}(\eta)|\le 2 |\sigma_{\xi}(s^{-1}(\eta))|, \ \  
|\hat\sigma_{\eta\eta}(\eta)|\le 8 |\sigma_{\xi\xi}(s^{-1}(\eta))|, \ \ 
|g_{2,\xi(\xi)}|\le 2 |\hat g_{2,\eta(s(\xi))}|,
$$
$$
\| \ldots\|_{L^2(d\eta)}\le \sqrt{2} \| \ldots\|_{L^2(d\xi)}.
$$
Therefore 
$$
|g_2(\xi)| \le\frac{1}{2\pi}\cdot[4 \max | \hat\sigma_{\eta}(\eta)| + 2 \| \hat\sigma\|_{L^2(d\eta)} ].
$$
Similarly:
$$
|g_{2,\xi}(\xi)| \le \frac{1}{\pi}\cdot[2 \cdot 8  \max | \sigma_{\xi\xi}(\xi)| +2 \cdot 2 \| \hat\sigma_{\xi}\|_{L^2(d\xi)} ].
$$
Let us proof the second part.

Assume that $|\eta|\le 2R$. Then
$$
\hat g_2(\eta)=\frac 1{2\pi i}\int\limits_{|\eta''-\eta|\le 3R}
\frac {\sigma(\eta'')}{\eta'' -\eta}
d\eta''=
\frac 1{2\pi i}\int\limits_{|\eta''-\eta|\le 3R}
\frac {\sigma(\eta'')-\sigma(\eta)}
{\eta'' -\eta}
d\eta''.
$$
We have:
$$
|\hat g_2(\eta)|\le \frac{1}{2\pi}\cdot 6R \cdot \max\limits_{\eta'}|\partial_{\eta'} \sigma(\eta')|\le  \frac{6R}{\pi}\cdot C_1.
$$
Consider now the case $|\eta|\ge 2R$. Then 
$$
\hat g_2(\eta)=\frac 1{2\pi i}\int\limits_{|\eta''|\le R}
\frac {\sigma(\eta'')}{\eta'' -\eta}
d\eta'', 
$$
$$
|\hat g_2(\eta)| \le
\frac 1{2\pi} \cdot \|\sigma(\eta)\|_{L^{\infty}(d\eta)} \cdot
\left|\int\limits_{\ |\eta''|\le R}
\frac {1}{\eta'' -\eta} d\eta'' \right|\le
$$
$$
\le \frac {C_0}{2\pi} \cdot \log\left(\frac{|\eta| +R}
{|\eta| -R}    \right)\le 
\frac {C_0}{2\pi} \cdot \log(3) \le \frac {C_0}{\pi}.
$$
For a finite support function 
$$
|\sigma(\xi)|\le R \cdot \|\sigma_{\xi}(\xi)\|_{L^{\infty}(d\xi)}, \ \ \mbox{i.e.} \ \ 
C_0 \le R C_1.
$$
The proof of the second formula is absolutely the same, but we take into account (\ref{eq:der-est-1}).

\textbf{Proof of Theorem~\ref{T:JOST}.}

In this part we always assume that $t>0$ is fixed, $\calD$ is an arbitrary fixed positive constant, $y<0$, 
$|y|$ is sufficiently large (more precisely, $|y|>64 \calD t$), $|x|\le\calD|y|$.

This proof consists of 3 steps:
\begin{enumerate}
\item We show that it is sufficient to obtain some  
$L^2(d\lambda)$ estimates on $\omega$ and $\omega_{\lambda}$.
\item  We show, that it is sufficient to estimate the first iteration of
$\omega$ and $\omega_{\lambda}$ in $L^2(d\lambda)$.
\item We estimate the first iteration of
$\omega$ and $\omega_{\lambda}$ in $L^2(d\lambda)$ for $y\rightarrow-\infty$.
\end{enumerate}

\textbf{Step 1.} 

From Lemma~\ref{Sobolev-embedding-1} it follows, that it is sufficient to prove 
the following:
$$
\|\omega(x,y,t,\lambda)\|_{L^2(d\lambda)} \cdot \|\omega_{\lambda}(x,y,t,\lambda)\|_{L^2(d\lambda)} \rightarrow 0 \ \ \mbox{as} \ \ y\rightarrow -\infty,
$$
uniformly in $x$ in the interval $|x|\le \calD |y|$.

\textbf{Step 2.} 

Let us recall that we use the following iteration procedure:
\begin{align}
\omega_{n+1}(x,y,t,\lambda)=
&-\chi_{-R}(x-\lambda y-\lambda^2 t + \omega_{n}(x,y,t,\lambda),\lambda)+\\
&+H_{\lambda}\left[\chi_{-I}(x-\lambda y-\lambda^2 t + \omega_{n}(x,y,t,\lambda),\lambda)\right],\nonumber\\
\omega_{n+1,\lambda}(x,y,t,\lambda)=&I_1+I_2 + y \cdot I_3, \hfill
\end{align}
where
\begin{align}
I_1=&-\chi_{-R,\lambda}(x-\lambda y-\lambda^2 t + \omega_{n}(x,y,t,\lambda),\lambda)+\\
&+H_{\lambda}\left[\chi_{-I,\lambda}(x-\lambda y-\lambda^2 t + \omega_{n}(x,y,t,\lambda),\lambda)\right]+\nonumber\\
&+2t\,\chi_{-R,\xi}(x-\lambda y-\lambda^2 t + \omega_{n}(x,y,t,\lambda),\lambda)
\cdot\lambda-\nonumber\\
&-2t\,H_{\lambda}\left[\chi_{-I,\xi}(x-\lambda y-\lambda^2 t + \omega_{n}(x,y,t,\lambda),\lambda)\cdot\lambda\right],\nonumber
\end{align}
\begin{align}
I_2=&-\chi_{-R,\xi}(x-\lambda y-\lambda^2 t + \omega_{n}(x,y,t,\lambda),\lambda)\cdot
\omega_{n,\lambda}(x,y,t,\lambda)+\nonumber\\
&+H_{\lambda}\left[\chi_{-I,\xi}(x-\lambda y-\lambda^2 t + \omega_{n}(x,y,t,\lambda),\lambda) \cdot \omega_{n,\lambda}(x,y,t,\lambda)
\right],\nonumber\\
\end{align}
\begin{align}
I_3=&\chi_{-R,\xi}(x-\lambda y-\lambda^2 t + \omega_{n}(x,y,t,\lambda),\lambda)-\\
&-H_{\lambda}\left[\chi_{-I,\xi}(x-\lambda y-\lambda^2 t + \omega_{n}(x,y,t,\lambda),\lambda)\right].\nonumber
\end{align}
It is convenient to write:
$$
I_3=I_{31}+I_{32}
$$
\begin{align}
&I_{31}=\chi_{-R,\xi}(x-\lambda y-\lambda^2 t + \omega_{n}(x,y,t,\lambda),\lambda)-
\chi_{-R,\xi}(x-\lambda y-\lambda^2 t,\lambda)-
\\
&-H_{\lambda}\left[\chi_{-I,\xi}(x-\lambda y-\lambda^2 t + \omega_{n}(x,y,t,\lambda),\lambda)-\chi_{-I,\xi}(x-\lambda y-\lambda^2 t,\lambda)
\right],\nonumber
\end{align}
\begin{align}
&I_{32}=I_{32}(x,y,t,\lambda)=\chi_{-R,\xi}(x-\lambda y-\lambda^2 t,\lambda)-
H_{\lambda}\left[\chi_{-I,\xi}(x-\lambda y-\lambda^2 t,\lambda)
\right].\nonumber
\end{align}
From (\ref{T:JOST:eq6}), (\ref{T:JOST:eq3}) we immediately obtain that there exists a constant $C_1>0$ such that 
\begin{equation}
\|I_1\|_{L^2(d\lambda)}<C_1 \ \ \mbox{for any} \ \ \omega(\lambda). 
\end{equation}
Using the same arguments as in Theorem~\ref{T:uniqueness-S2} we immediately obtain
\begin{equation}
\|I_2\|_{L^2(d\lambda)}<\frac{1}{2}  \|\omega_{n,\lambda}\|_{L^2(d\lambda)}.
\end{equation}
From  (\ref{T:JOST:eq5}) we immediately obtain that there exists $C_2>0$ such that
\begin{equation}
\|I_{31}\|_{L^2(d\lambda)}< C_2  \|\omega_{n}\|_{L^2(d\lambda)}.
\end{equation}
We also know that
\begin{equation}
\|\omega_{n}\|_{L^2(d\lambda)}\le 2  \|\omega_{1}\|_{L^2(d\lambda)}.
\end{equation}
Combining all these estimates we obtain:
\begin{equation}
\|\omega_{n+1}(x,y,t,\lambda)\|_{L^2(d\lambda)}<  C_1+ \frac{1}{2}  \|\omega_{n,\lambda}\|_{L^2(d\lambda)} + 
2 |y|\cdot C_2  \|\omega_{1}\|_{L^2(d\lambda)}+ |y|\cdot \|I_{32}\|_{L^2(d\lambda)}.
\end{equation}

\textbf{Therefore, to prove Theorem~\ref{T:JOST}, it is sufficient to show that} 
\begin{equation}
\|\omega_{1}\|_{L^2(d\lambda)} = o\left(\frac{1}{\sqrt{|y|}} \right), \ \ 
\|I_{32}\|_{L^2(d\lambda)}=  o\left(\frac{1}{\sqrt{|y|}} \right) 
\ \ \mbox{as} \ \ y\rightarrow-\infty,
\end{equation}
\textbf{uniformly in $x$ for $|x|\le\calD |y|$, where $\calD$ is an arbitrary positive 
constant.}

\textbf{Step 3.}

The proof of both estimates are absolutely similar; moreover the second one is a 
little easier from a technical point of view. Let us estimate $\omega_1$:
$$
\omega_{1}(x,y,t,\lambda)= H_{\lambda}\left[\chi_{-I}( x -\lambda y-\lambda^2 t ,\lambda)\right]- \chi_{-R}(\tau-\lambda y -\lambda_0^2 t,\lambda).
$$
It is convenient to represent $\chi_{-}(\xi,\lambda)$ as a sum of three functions:
\beq
\chi_{-}(\xi,\lambda)=\chi^{(1)}_{-}(\xi,\lambda)+\chi^{(2)}_{-}(\xi,\lambda)+
\chi^{(3)}_{-}(\xi,\lambda)
\eeq
\beq
\chi^{(1)}_{-}(\xi,\lambda)=\left\{
\begin{array}{ll}\chi_{-}(\xi,\lambda), & |\lambda|\le 4\calD \\0, &
|\lambda| > 4\calD, \end{array}\right.
\eeq
\beq
\chi^{(2)}_{-}(\xi,\lambda)=\left\{
\begin{array}{ll}\chi_{-}(\xi,\lambda), & 4\calD < |\lambda|\le |y|/4t \\0, &
\mbox{otherwise}, \end{array}\right.
\eeq
\beq
\chi^{(3)}_{-}(\xi,\lambda)=\left\{
\begin{array}{ll}\chi_{-}(\xi,\lambda), & |\lambda| > |y|/4t \\0, &
|\lambda| \le |y|/4t. \end{array}\right.
\eeq
From  (\ref{T:JOST:eq1}) it follows immediately that there exists a constant $C_3>0$ such 
that
\beq
\|\chi^{(3)}_{-}\|_{L^2(d\lambda)}\le \frac{C_3}{|y|^{3/2}}.
\eeq
If $4D < |\lambda|\le |y|/4t$, $|x|\le \calD |y|$  then $|x-\lambda y -\lambda^2t|>
|\lambda y|/2$, and 
\beq
|\chi^{(2)}_{-}(\xi,\lambda)|\le\frac{2C}{|\lambda|\cdot|y|}, 
\eeq
\beq
\|\chi^{(2)}_{-}\|_{L^2(d\lambda)}\le \frac{4C}{|y|} \sqrt{\int\limits_{4\calD}^{\infty}\frac{d\lambda}{\lambda^2}} 
\eeq
Let us denote by $\omega^{(1)}$ the function:
$$
\omega^{(1)}(x,y,t,\lambda)= H_{\lambda}\left[ \chi^{(1)}_{-I}( x -\lambda y-\lambda^2 t ,\lambda)\right]- \chi^{(1)}_{-R}(x-\lambda y -\lambda^2 t,\lambda).
$$
We have shown that
$$
\| \omega^{(1)}-\omega_{1}\|_{L^2(d\lambda)}=O\left(\frac{1}{|y|}\right);
$$
therefore it is sufficient to estimate  $\| \omega^{(1)}\|_{L^2(d\lambda)}$. We have
\begin{equation}
\label{T:JOST:omega1}
\| \omega^{(1)}\|_{L^2(d\lambda)}\le \| \omega^{(1)}\|_{L^2(d\lambda),|\lambda|\le 2\calD} +
\|\chi^{(1)}_{-R}\|_{L^2(d\lambda),2\calD\le|\lambda|\le 4\calD}+ 
\|H_{\lambda}\left[\chi^{(1)}_{-I}\right]\|_{L^2(d\lambda),|\lambda|\ge2\calD}.
\end{equation}
For sufficiently large $|y|$ and $\frac{3}{2}\calD\le|\lambda|\le4\calD$ we have
\beq
\label{T:JOST:chi1}
|\chi^{(1)}_{-}(x-\lambda y-\lambda^2 t,\lambda)|\le\frac{4C}{\calD|y|},
\eeq
and
\beq
\|\chi^{(1)}_{-R}\|_{L^2(d\lambda),2\calD\le|\lambda|\le4\calD}\le \frac{16C}{|y|}.
\eeq
Let us estimate the $L^2$-norm of $H_{\lambda}
\left[\chi^{(1)}_{-I}(x-\lambda y-\lambda^2 t,\lambda)\right]$
on the interval $|\lambda|>2\calD$. We have
\begin{equation}
H_{\lambda}\left[\chi^{(1)}_{-I}(x-\lambda y-\lambda^2 t,\lambda)\right]=
\frac{1}{\pi}\int_{-4\calD}^{4\calD} 
\frac{ \chi^{(1)}_{-I}(x-\mu y-\mu^2 t,\mu) d\mu}{\lambda-\mu}=I_1(\lambda)+I_2(\lambda),
\end{equation}
where
\begin{equation}
I_1(\lambda)=
\frac{1}{\pi}\int_{-\frac{3}{2}\calD}^{\frac{3}{2}\calD} 
\frac{ \chi^{(1)}_{-I}(x-\mu y-\mu^2 t,\mu) d\mu}{\lambda-\mu},
\end{equation}
\begin{equation}
I_2(\lambda)=
\frac{1}{\pi}\int\limits_{\frac{3}{2}\calD\le|\mu|\le4\calD} 
\frac{ \chi^{(1)}_{-I}(x-\mu y-\mu^2 t,\mu) d\mu}{\lambda-\mu}.
\end{equation}
From (\ref{T:JOST:chi1}) it follows that 
\begin{equation}
\|I_2(\lambda)\|_{L^2(d\lambda)}\le 
\|\chi^{(1)}_{-I}\|_{L^2(d\lambda),\frac{3}{2}\calD\le|\lambda|\le 4\calD} \le \frac{20C}{|y|}.
\end{equation}
For $|\lambda|>2\calD$ we have
\begin{align}
|I_1(\lambda)|\le&\frac{1}{\pi}\frac{1}{|\lambda|-\frac{3}{2}\calD} \int_{-\frac{3}{2}\calD}^{\frac{3}{2}\calD} 
|\chi^{(1)}_{-I}(x-\mu y-\mu^2 t,\mu)| d\mu\le\\
&\le\frac{1}{\pi}\frac{1}{|\lambda|-\frac{3}{2}\calD} \int_{-\frac{3}{2}\calD}^{\frac{3}{2}\calD} 
\frac{C d\mu}{1+|x-\mu y -\mu^2 t|}\le \frac{C_4+C_5\log|y|}
{|y|(|\lambda|-\frac{3}{2}\calD)}.
\end{align}
\begin{equation}
\|I_1(\lambda)\|_{L^2(d\lambda),|\lambda|\ge 2\calD}\le 2\sqrt{2}\cdot \frac{C_4+C_5\log|y|}
{|y|}.
\end{equation}

To complete the proof, we have to estimate 
$\omega^{(1)}(x,y,t,\lambda)$  in the interval  $-2\calD\le\lambda\le2\calD$.

For $y<0$ the function  $\chi_{-}(x-\mu y - \lambda^2t,\lambda)$ is holomorphic in $\mu$ in 
the lower half-plane; therefore
\begin{equation}
\chi^{(1)}_{-R}(x-\lambda y-\lambda^2 t,\lambda)= H_{\mu}\left[\chi^{(1)}_{-I}(x-\mu y-\lambda^2 t,\lambda)\right]_{\mu=\lambda},
\end{equation}
and 
\begin{align}
&\omega^{(1)}(x,y,t,\lambda)=\\
&=-\chi^{(1)}_{-R}(x-\lambda y-\lambda^2 t,\lambda)+ H_{\lambda}\left[\chi^{(1)}_{-I}(x-\lambda y-\lambda^2 t,\lambda)\right]=\nonumber\\
&=\frac{1}{\pi}\int_{-4\calD}^{4\calD} 
\frac{ \chi^{(1)}_{-I}(x-\mu y-\mu^2 t,\mu) d\mu}{\lambda-\mu}-
\frac{1}{\pi}\int_{-\infty}^{\infty} 
\frac{ \chi^{(1)}_{-I}(x-\mu y-\lambda^2 t,\lambda) d\mu}{\lambda-\mu}=\nonumber\\
& =\frac{1}{\pi}\int_{-4\calD}^{4\calD} 
\frac{ \chi^{(1)}_{-I}(x-\mu y-\mu^2 t,\mu) -\chi^{(1)}_{-I}(x-\mu y-\lambda^2 t,\lambda)  }{\lambda-\mu} d\mu + I_3,\nonumber
\end{align}
where
\begin{equation}
I_3=- \frac{1}{\pi}\int\limits_{|\mu|>4\calD} 
\frac{ \chi^{(1)}_{-I}(x-\mu y-\lambda^2 t,\lambda) d\mu}{\lambda-\mu}.
\end{equation}
If $|y|>4\calD t$, then 
\begin{equation}
|I_3|\le\frac{1}{\pi}\int\limits_{|\mu|>4\calD} 
\frac{4C}{|y||\mu|^2 }d\mu= \frac{1}{|y|} \frac{2C}{\pi\calD}.
\end{equation}
\begin{equation}
\frac{1}{\pi}\int_{-4\calD}^{4\calD} 
\frac{ \chi^{(1)}_{-I}(x-\mu y-\mu^2 t,\mu) -\chi^{(1)}_{-I}(x-\mu y-\lambda^2 t,\lambda)  }{\lambda-\mu} d\mu =I_4+I_5,
\end{equation}
where
\begin{equation}
I_4=\frac{1}{\pi}\int_{-4\calD}^{4\calD} 
\frac{ \chi^{(1)}_{-I}(x-\mu y-\mu^2 t,\lambda) -\chi^{(1)}_{-I}(x-\mu y-\lambda^2 t,\lambda)}{\lambda-\mu} d\mu, 
\end{equation}
\begin{equation}
I_5=\frac{1}{\pi}\int_{-4\calD}^{4\calD} 
\frac{ \chi^{(1)}_{-I}(x-\mu y-\mu^2 t,\mu) -\chi^{(1)}_{-I}(x-\mu y-\mu^2 t,\lambda)}{\lambda-\mu} d\mu. 
\end{equation}
We see, that
\begin{equation}
I_4=\frac{1}{\pi}\int_{-4\calD}^{4\calD} t(\lambda+\mu)\chi^{(1)}_{-I,\xi}(\hat\xi,\lambda) d\mu, \ \ \mbox{where} \ \ \hat\xi\in[x-\mu y-\lambda^2t, x-\mu y-\mu^2t],
\end{equation}
Denote:
\begin{equation}
I_4=I_{41}+I_{42},
\end{equation}
where
\begin{align}
&I_{41}=\frac{1}{\pi}\int\limits_{|\mu|\le4\calD, |x-\mu y|>64 D^2t} t(\lambda+\mu)\chi^{(1)}_{-I,\xi}(\hat\xi,\lambda) d\mu\\
&I_{42}=\frac{1}{\pi}\int\limits_{|\mu|\le4\calD, |x-\mu y|\le 64 D^2t}t(\lambda+\mu)\chi^{(1)}_{-I,\xi}(\hat\xi,\lambda) d\mu,\nonumber
\end{align}
\begin{equation}
|I_{42}|\le \frac{1024 C \calD^3 t^2}{\pi|y|},
\end{equation}
\begin{equation}
|I_{41}|\le \frac{8\calD t}{\pi}\int\limits_{|\mu|\le4\calD}\frac{C}{1+\frac{|x-\mu y|^2}{4}}d\mu\le
\frac{16C\calD t}{\pi|y|}\int\limits_{-\infty}^{\infty}\frac{1}{1+\mathring\mu^2}d\mathring\mu= \frac{16C\calD t}{|y|}.
\end{equation}
Analogously, 
\begin{equation}
I_5=\frac{1}{\pi}\int_{-4\calD}^{4\calD}  \chi^{(1)}_{-I,\lambda}(x-\mu y-\mu^2 t,\hat\mu) d\mu, \ \ \mbox{where} \ \ |\hat\mu|\le 4\calD, 
\end{equation}
\begin{equation}
I_5=I_{51}+I_{52},
\end{equation}
\begin{align}
&I_{51}=\frac{1}{\pi}\int\limits_{|\mu|\le4\calD, |x-\mu y|>64 D^2t} 
\chi^{(1)}_{-I,\lambda}(x-\mu y -\mu^2 t,\hat\mu) d\mu\\
&I_{52}=\frac{1}{\pi}\int\limits_{|\mu|\le4\calD, |x-\mu y|\le 64 D^2t} 
\chi^{(1)}_{-I,\lambda}(x-\mu y-\mu^2 t,\hat\mu) d\mu,\nonumber
\end{align}
\begin{equation}
|I_{52}|\le \frac{128 C \calD^2 t}{\pi|y|},
\end{equation}
\begin{equation}
|I_{51}|\le \frac{1}{\pi}\int\limits_{|\mu|\le4\calD}\frac{C}{1+\frac{|x-\mu y|}{2}}d\mu\le
\frac{4C}{\pi|y|}\int\limits_{0}^{\frac{5}{2}\calD|y|}\frac{1}{1+\mathring\mu}d\mathring\mu=
\frac{4C}{\pi|y|}\log\left(\frac{5}{2}\calD|y|\right).
\end{equation}

We have shown that there exist positive constants $C_6=C_6(\calD)$, $C_7=C_7(\calD)$ such
that
\begin{equation}
\|\omega^{(1)} \|_{L^2(d\lambda)}\le\frac{ C_6+C_7 \log|y|}{|y|}.
\end{equation}
Analogously, there exists a constant $C_8=C_8(\calD)$ such that 
\begin{equation}
\|I_{32}(x,y,t,\lambda)\|_{L^2(d\lambda)}\le\frac{ C_8}{|y|}
\end{equation}
($\chi_{-,\,\xi}(\xi,\lambda)$,  $\chi_{-,\,\xi\lambda}(\xi,\lambda)$ decay at $|\xi|\rightarrow\infty$ as $1/|\xi|^2$, therefore we have no logarithmic terms).

The proof is completed.

{\bf Acknowledgments}. An essential part of this work was made during the visit of the three authors to the Centro Internacional de Ciencias in Cuernavaca, Mexico in November-December 2012. The first author was also partially supported by the Russian Foundation for Basic Research, grant 13-01-12469 ofi-m2, Russian Federation Government grant No~2010-220-01-077, by the program ``Leading scientific schools'' (grant NSh-4833.2014.1) and by the program  ``Fundamental problems of nonlinear dynamics''. The third author was  partially supported by NSC 101-2115-M-001-002. We would also like to thank L. Pizzocchero for a useful information.


\begin{thebibliography}{9}

\bibitem{ABF}  M. J. Ablowitz, D. Bar Yaacov and A. S. Fokas, ``On the Inverse Scattering Transform for the Kadomtsev-Petviashvili Equation'', Stud. Appl. Math., {\bf 69}, (1983), 135-143.

\bibitem{AC} M. J. Ablowitz and P. A. Clarkson, \textit{Solitons, nonlinear evolution equations and Inverse Scattering}, London Math. Society Lecture Note 1991.

\bibitem{AS} M. J. Ablowitz and H. Segur \textit{Solitons and the Inverse Scattering Transform} SIAM 1981.

\bibitem{Arnold} V. I. Arnold, \textit{Ordinary Differential Equations}, Springer, 1992.

\bibitem{Bogdanov1} L. V. Bogdanov ``On a class of reductions of Manakov-Santini hierarchy connected with the interpolating system'', J. Phys. A: Math. Theor. {\bf 43}, 
(2010), 115206 (11pp).

\bibitem{Bogdanov2} L. V. Bogdanov ``Interpolating differential reductions of multidimensional integrable hierarchies'', Theor. Math. Phys., {\bf 167} (3), (2011), 705-713.

\bibitem{BDM} L. Bogdanov, V. Dryuma and S. V. Manakov: ``Dunajski generalization of the second heavenly equation: dressing method and the hierarchy'', J. Phys. A: Math. Theor. 40, (2007), 14383-14393, doi:10.1088/1751-8113/40/48/005.

\bibitem{BK} L. V. Bogdanov and B. G. Konopelchenko, ``On the $\bar\partial$-dressing method applicable to heavenly equation''; Phys. Lett. A {\bf 345}, (2005) 137-143. 

\bibitem{BF} C. Boyer and J. D. Finley, ``Killing vectors in self-dual, Euclidean Einstein spaces'', J. Math. Phys. {\bf 23} (1982), 1126-1128.

\bibitem{CD} F. Calogero and A. Degasperis {\it Spectral Transform and Solitons} North-Holland Publishing Company, 1982.

\bibitem{DS} A. Davey and K. Stewartson, ``On Three-Dimensional Packets of Surface Waves'', Proc. R. Soc. A, {\bf 338}, (1974), 101-110. 

\bibitem{Duna1} M. Dunajski, \textit{The nonlinear graviton as an integrable system}, PhD Thesis, Oxford University, 1998. 

\bibitem{Duna} M. Dunajski, ``A class of Einstein-Weyl spaces associated to an integrable system of hydrodinamic type'', J. Geom. Phys. {\bf 51}  (2004), 126-137.
  
\bibitem{Dunajj} M. Dunajski, ``Interpolating dispersionless integrable system'', J. Phys. A: Math. Theor., {\bf 41} (2008), 315202 (9pp),  doi:10.1088/1751-8113/41/31/315202

\bibitem{DunaFer} M. Dunajski, E. Ferapontov, B. Kruglikov, ``On the Einstein-Weyl and conformal self-duality equations'',  arXiv:1406.0018 [nlin.SI].

\bibitem{DuMa1} M. Dunajski and L. J. Mason, ``Hyper-K\"ahler hierachies and their twistor theory'', Comm. Math. Phys. {\bf 213}, (2000), 641-672.  

\bibitem{DuMa2} M. Dunajski and L. J. Mason, ``Twistor theory of hyper-K\"ahler metrics with hidden symmetries'', J. Math. Phys., {\bf 44},(2003), 3430-3454. 

\bibitem{DMT} M. Dunajski, L. J. Mason and P. Tod, ``Einstein-Weyl geometry, the dKP equation and twistor theory'', J. Geom. Phys. {\bf 37} (2001), 63-93.

\bibitem{DT} M. Dunajski and K. P. Tod, ``Einstein-Weyl spaces and dispersionless Kadomtsev-Petviashvili equation from Painlev\'e I and II''; arXiv:nlin.SI/0204043. 

\bibitem{EE} W. E, and B. Engquist: ``Blowup of solutions of the unsteady Prandtl's equation'', Communications on Pure and Applied Mathematics, {\bf 50}, Issue 12 (1997), 1287-1293.

\bibitem{Ferapontov} E. V. Ferapontov and K. R. Khusnutdinova: ``On integrability of (2+1)-dimensional quasilinear systems'', Comm. Math. Phys. {\bf 248} (2004) 187-206.

\bibitem{FP} J. D. Finley and J. F. Plebanski: ``The classification of all ${\cal K}$ spaces admitting a Killing vector'', J. Math. Phys. {\bf 20}, (1979), 1938.

\bibitem{FA}   A. S. Fokas and M. J. Ablowitz,  ``On the Inverse Scattering of the Time Dependent Schr\"odinger Equations and the Associated KPI Equation'', 
Stud. in Appl. Math, 69, (1983), 211-228.

\bibitem{Ga66} F. D. Gakhov, \textit{Boundary value problems}, 1966, Translation edited by I. N. Sneddon Pergamon Press, Oxford-New York-Paris; Addison-Wesley Publishing Co., Inc., Reading, Mass. London. 

\bibitem{GGKM} C. S. Gardner, J. M. Greene, M. D. Kruskal and R. M. Miura, ``Method for Solving the Korteweg-deVries Equation'', Phys. Rev. Lett., 
{\bf 19}, (1967), 1095-1097.

\bibitem{GD} J. D. Gegenberg and A. Das, ``Stationary Riemaniann space-times with self-dual curvature'', 
Gen. Rel. Grav. {\bf 16} (1984), 817-829.

\bibitem{GK} I. Ts. Gokhberg, N. Ya. Krupnik, ``Norm of the Hilbert transformation in the $L_p$ space'',  Funct. Anal. Appl., \textbf{2}, Issue 2 (1968), 
180-181. 

\bibitem{GS} P. G. Grinevich and P. M. Santini: ``Holomorphic eigenfunctions of the vector field associated with the dispersionless Kadomtsev - Petviashvili equation'', arXiv:1111.4446. J. Differential Equations, {\bf 255}, 7, 1469-1491 (2013). http://dx.doi.org/10.1016/j.jde.2013.05.010 

\bibitem{G-M-MA} F. Guil, M. Manas and L. Martinez Alonso, ``On twistor solutions of the dKP equation'', J. Phys. A:Math. Gen. {\bf 36} (2003) 6457-6472.

\bibitem{J} P. E. Jones and K. P. Tod, ``Minitwistor spaces and Einstein-Weyl spaces'', Class. Quantum Grav. 
{\bf 2} (1985), 565-577.

\bibitem{H} N. J. Hitchin, ``Complex manifolds and Einstein's equations'', in {\it Twistor Geometry and 
Nonlinear Systems}, H. D. Doebner and T. Weber (eds), Lecture Notes in Mathematics, vol. 970 (Springer-Verlag 
1982).

\bibitem{KP} B. B. Kadomtsev and V. I. Petviashvili,  ``On the stability of solitary waves in weakly dispersive media'', Sov. Phys. Dokl., {\bf 15}, (1970), 539-541. 

\bibitem{KG} Y. Kodama and J. Gibbons, ``Integrability of the dispersionless KP hierarchy'', 
Proc. 4th Workshop on Nonlinear and Turbulent Processes in Physics, World Scientific, Singapore 1990.

\bibitem{KM1} B. Konopelchenko and F. Magri, ``Coisotropic deformations of associative algebras and dispersionless integrable hierarchies'', Comm. Math. Phys. {\bf 274}, (2007), 627-658. 

\bibitem{KM2} B. Konopelchenko and F. Magri, ``Dispersionless integrable equations as coisotropic deformations. Extensions and reductions'', Theoretical and Mathematical Physics, {\bf 151}(3), (2007), 803-819.
  
\bibitem{K-MA-R} B. Konopelchenko, L. Martinez Alonso and O. Ragnisco, ``The $\bar\partial$-approach for the dispersionless KP hierarchy'', J.Phys. A: Math. Gen. {\bf 34} 
(2001) 10209-10217.

\bibitem{KdV} D. J. Korteweg and G. de Vries, ``On the change of form of long waves advancing in a rectangular canal, and on a new type of long stationary waves'', Phil. Mag. \textbf{39} (1895), 422--443.

\bibitem{Kri0} I. M. Krichever, ``Method of averaging for two-dimensional ``integrable'' equations'', Functional Analysis and Its Applications, {\bf 22} (1988), 200-213.

\bibitem{Kri1} I. M. Krichever, ``The dispersionless Lax equations and topological minimal models'', Comm. Math. Phys. {\bf 143} (1992), no. 2, 415-429. 

\bibitem{Kri2} I. M. Krichever, ``The $\tau$-function of the universal Witham hierarchy, matrix models and topological field theories'', Comm. Pure Appl. Math. {\bf 47}, 437-475 (1994). 

\bibitem{KMZ} I. Krichever, A. Marshakov and A. Zabrodin, ``Integrable structure of the Dirichlet boundary problem in multiply-connected domains'',  
Comm. Math. Phys.  {\bf 259}  (2005),  no. 1, 1-44.

\bibitem{Kpr1} Private communication by V.E. Kuznetsov, December 2013. 

\bibitem{LBW} S.-Y. Lee, E. Bettelheim, P. Wiegmann, ``Bubble break-off in Hele-Shaw flows-singularities 
and integrable structures''  Phys. D,  {\bf 219}  (2006),  no. 1, 22-34.  

\bibitem{Lin} C. C. Lin, E. Reissner,  and H.S. Tsien,  ``On two-dimensional non-steady motion of a slender
body in a compressible fluid''. Journal of Mathematical Physics, 27, (1948). 220-231.

\bibitem{Manakov1} S. V. Manakov, ``The inverse scattering transform for the time - dependent Schr\"odinger operator and Kadomtsev-Petviashvili equation'', Physica {\bf 3D}, 420-427 (1981).

\bibitem{MS0} S. V. Manakov and P. M. Santini 2005 Inverse Scattering Problem for Vector Fields and the Heavenly Equation {\it Preprint arXiv:nlin/0512043}.

\bibitem{MS1} V. Manakov and P. M. Santini: ``Inverse scattering problem for vector fields and the Cauchy problem for the heavenly equation'', Physics Letters A {\bf 359} (2006) 613-619. http://arXiv:nlin.SI/0604017.

\bibitem{MS2} S. V. Manakov and P. M. Santini: ``The Cauchy problem on the plane for the dispersionless Kadomtsev-Petviashvili equation''; JETP Letters, {\bf 83}, No 10, 462-466 (2006). http://arXiv:nlin.SI/0604016.

\bibitem{MS3} S. V. Manakov and P. M. Santini: ``A hierarchy of integrable PDEs in $2+1$ dimensions associated with $1$ - dimensional vector fields''; Theor. Math. Phys. {\bf 152}(1), 1004-1011 (2007).

\bibitem{MS4} S. V. Manakov and P. M. Santini: ``On the solutions of the dKP equation: the nonlinear Riemann-Hilbert problem, longtime behaviour, implicit solutions and wave breaking''; J. Phys. A: Math. Theor. {\bf 41} (2008) 055204 (23pp).

\bibitem{MS5} S. V. Manakov and P. M. Santini: ``On the solutions of the second heavenly and Pavlov equations'', J. Phys. A: Math. Theor. {\bf 42} (2009) 404013 (11pp). doi: 10.1088/1751-8113/42/40/404013. arXiv:0812.3323.

\bibitem{MS6} S. V. Manakov and P. M. Santini: ``The dispersionless 2D Toda equation: dressing, Cauchy problem, longtime behaviour, implicit solutions and wave breaking'', J. Phys. A: Math. Theor. {\bf 42} (2009) 095203 (16pp).

\bibitem{MS7} S. V. Manakov and P. M. Santini: ``Solvable vector nonlinear Riemann problems, exact implicit solutions of dispersionless PDEs and wave breaking'', J. Phys. A: Math. Theor. 44 (2011) 345203 (19pp), doi:10.1088/1751-8113/44/34/345203. arXiv:1011.2619. 

\bibitem{MS8} S. V. Manakov and P. M. Santini: ``On the dispersionless Kadomtsev-Petviashvili equation in n+1 dimensions: exact solutions, the Cauchy problem for small initial data and wave breaking'', J. Phys. A: Math. Theor. 44 (2011) 405203 (15pp), doi:10.1088/1751-8113/44/40/405203. arXiv:1001.2134.

\bibitem{MS9} S. V. Manakov and P. M. Santini: ``Wave breaking in solutions of the dispersionless Kadomtsev-Petviashvili equation at finite time'', Theor. Math. Phys. {\bf 172}(2), (2012), 1118-1126.

\bibitem{MS10} S. V. Manakov and P. M. Santini: ``Integrable dispersionless PDEs arising as commutation condition of pairs of vector fields'' Proceedings of the conference PMNP 2013, IOP Conference Series. arXiv:1312.2740.

\bibitem{MZ} S. V. Manakov and V. E. Zakharov, ``Three-dimensional model of relativistic-invariant field theory, integrable by the inverse scattering transform''; Letters in Mathematical Physics {\bf 5}, (1981) 247-253.

\bibitem{MAM} L. Martinez Alonso and E. Medina: ``Regularisation of Hele-Shaw flows, multiscaling expansions and 
the Painlev\'e I equation'', arXiv:0710.3731.

\bibitem{MA-S} L. Martinez Alonso and A. B. Shabat, ``Towards a theory of differential constraints of a hydrodynamic hierarchy'' J. Nonlinear Math. Phys 10(2) (2003) 229. 

\bibitem{MWZ} M. Mineev-Weinstein, P. Wigmann and A. Zabrodin, ``Integrable Structure of Interface Dynamics'', Phys. Rev. Lett. {\bf 84}, (2000), 5106.

\bibitem{NNS} F. Neyzi, Y. Nutku and M. B. Sheftel, ``A multi-hamiltonian structure of the Plebanski's second heavenly equation''; J. Phys. A: Math. Gen. {\bf 38} (2005), 
8473-8485.

\bibitem{Pavlov} M. V. Pavlov: ``Integrable hydrodynamic chains'', J. Math. Phys. {\bf 44} (2003) 4134-4156.

\bibitem{Pic} S. Pichorides, ``On the best values of the constants in the theorem of M. Riesz, Zygmund and Kolmogorov'',  Studia Mathematica \textbf{44}, Issue: 2 (1972), 165-179.

\bibitem{Plebanski} J. F. Plebanski, ``Some solutions of complex Einstein equations'', J. Math. Phys. {\bf 16}, (1975), 2395-2402.

\bibitem{S70} E. M. Stein, \textit{Singular integrals and differentiability properties of functions}, 1970, no. 30, Princeton University Press, Princeton, N.J..

\bibitem{Taka1} K. Takasaki, ``Area preserving diffeomorphisms and nonlinear integrable systems'', 1991, in ``Turku 1991, Proceedings, Topological and geometrical methods in field theory 383''.

\bibitem{TT1} K. Takasaki and T. Takebe, ``SDiff(2) Toda equation -- hierarchy, tau function and symmetries'', Lett. Math. Phys. {\bf 23}, (1991), 205-214.

\bibitem{TT2} K. Takasaki K and T. Takebe, ``SDiff(2) KP hierarchy'', A. Tsuchiya, T. Eguchi and T. Miwa (eds.), {\it Infinite Analysis\/}, Adv. Ser. Math. Phys. 16 (World Scientific, Singapore, 1992), part B, 889-922.

\bibitem{TT3} K. Takasaki and T. Takebe, ``Integrable hierarchies and dispersionless limit'', Rev. Math. Phys. {\bf 7}, (1995), 743.

\bibitem{Timman} R. Timman, ``Unsteady motion in transonic flow'', Symposium Transsonicum, Aachen 1962. Ed. K. Oswatitsch,  Springer 394-401.

\bibitem{V62} I. N. Vekua, \textit{Generalized analytic functions}, 1962, Pergamon Press, London-Paris-Frankfurt; Addison-Wesley Publishing Co., Inc., Reading, Mass..


\bibitem{Ward} R. S. Ward, ''Einstein-Weyl spaces and SU($\infty$) Toda fields'', Class. Quantum Grav. 
{\bf 7} (1990) L95-L98.

\bibitem{WZ} P. Wigmann and A. Zabrodin, ``Conformal Maps and Integrable Hierarchies'', Comm. Math. Phys. {\bf 213}, (2000), 523-528.

\bibitem{Zak} V. E. Zakharov: ``Integrable systems in multidimensional spaces'', Lecture Notes in Physics, Springer-Verlag, Berlin {\bf 153} (1982), 190-216.

\bibitem{Zakharov} E. Zakharov, ``Dispersionless limit of integrable systems in 2+1 dimensions'', in {\it Singular Limits of Dispersive Waves}, edited by N.M.Ercolani et al., Plenum Press, New York, 1994.

\bibitem{ZKpr1} Private communication by V.E. Zakharov and V.E. Kuznetsov. 

\bibitem{ZMNP} V. E. Zakharov, S. V. Manakov, S. P. Novikov and L. P. Pitaevsky, \textit{Theory of solitons}, 1984, Plenum Press, New York. 

\bibitem{ZS} V. E. Zakharov and A. B. Shabat, ``Integration of nonlinear equations of mathematical physics by the method of inverse scattering. II'', 
Functional Anal. Appl. {\bf 13}, (1979), 166-174.

\bibitem{ZSNLS} V. E. Zakharov and A. B. Shabat, ``Interaction between solitons in a stable medium'', Sov Phys {\it JETP} {\bf 37}, (1973) 823-828. 

\bibitem{ZK} E. A. Zobolotskaya and R. V. Kokhlov, ``Quasi - plane waves in the nonlinear acoustics of confined beams'',  Sov. Phys. Acoust. {\bf 15}, n. 1, (1969) 35-40. 

\end{thebibliography}
\end{document}